\documentclass[runningheads]{llncs}
\usepackage[T1]{fontenc}
\usepackage{graphicx}
\usepackage{orcidlink}
\usepackage{amsmath,amsfonts}
\usepackage{subcaption}
\usepackage{bbm}
\usepackage{mathtools}

\usepackage{hyperref}
\usepackage{cleveref}

\DeclareMathOperator{\portal}{portal}

\DeclareMathOperator*{\argmin}{argmin}
\DeclareMathOperator*{\argmax}{argmax}
\newcommand{\gate}[2]{\ensuremath{P_{#2}^{#1}}}
\newcommand{\boundarynode}[2]{\ensuremath{b_{#2}^{#1}}}

\spnewtheorem{observation}{Observation}{\bfseries}{\itshape}

\allowdisplaybreaks

\begin{document}

\title{Logarithmic-Time Geodesically Convex Decomposition in Programmable Matter\thanks{This work has been supported by the DFG Project SCHE 1592/10-1.}} 

\titlerunning{Geodesically Convex Decomposition in Programmable Matter}

\author{Henning Hillebrandt \inst{1}\orcidlink{0009-0007-7029-0014}\and
Andreas Padalkin\inst{1}\orcidlink{0000-0002-4601-9597}\and
Christian Scheideler\inst{1}\orcidlink{0000-0002-5278-528X}\and
Daniel Warner\inst{1}\orcidlink{0000-0002-9423-6094}\and
Julian Werthmann \inst{1}\orcidlink{0000-0002-5110-5625}}

\authorrunning{H. Hillebrandt, A. Padalkin, C. Scheideler, D. Warner and J. Werthmann}

\institute{
Paderborn University, Paderborn, Germany\\
\email{\{hhilleb, padalkin, scheidel, dwarner, jwerth\}@mail.upb.de}}

\maketitle

\begin{abstract}
The decomposition of complex structures into ``simpler'' substructures is a powerful technique with a wide range of applications. We study the computation of decompositions in the context of programmable matter. The \emph{amoebot model} is a well-established model for programmable matter, which places $n$ tiny robots called amoebots on the triangular grid. We consider the \emph{reconfigurable circuit extension} of the geometric amoebot model, which allows amoebots to interconnect via so-called circuits. Amoebots can then instantaneously transmit simple beeps to all amoebots connected by the same circuit. Using reconfigurable circuits, previous papers have described a linear-time triangulation algorithm, and a logarithmic-time decomposition algorithm into so-called tunnel regions. Both algorithms only work on a restricted class of amoebot structures. In this paper, we define a decomposition into $O(|\mathcal H|)$ simple, geodesically convex regions for arbitrary amoebot structures, and show how it can compute such a decomposition in $O(\log n)$ rounds, where $|\mathcal H|$ denotes the number of holes in the amoebot structure. As a byproduct, we also improve the global maxima algorithm of Padalkin \emph{et al.} (Nat. Comput., 2024) for special cases and with that also their spanning tree algorithm to $O(\log n)$ rounds w.h.p.

\keywords{Programmable matter, amoebot model, reconfigurable circuits, decomposition} 
\end{abstract}

\section{Introduction}

The decomposition of complex structures into ``simpler'' substructures is a powerful technique with a wide range of applications.
For example, decomposing a nonconvex polygon into convex subpolygons allows us to apply algorithms for collision detection of convex polygons (e.g., \cite{DBLP:journals/cvgip/ORourkeCON82a}) on any polygon.
Other applications include
pattern recognition \cite{DBLP:journals/tc/FengP75},
Minkowski sum computation \cite{DBLP:journals/comgeo/AgarwalFH02},
motion planning \cite{DBLP:journals/ijcga/HertL98},
origami folding \cite{DBLP:journals/comgeo/DemaineDM00},
routing \cite{DBLP:journals/tcs/CoyCSSW24},
skeleton computation \cite{DBLP:conf/sma/LienKA06},
shape recognition \cite{DBLP:journals/jcb/FeldmannPSD22},
and
shortest path computation \cite{DBLP:conf/podc/PadalkinS24}.
These applications may require different properties for the simpler structures.
Typical properties are
simplicity (i.e., hole-freeness) \cite{DBLP:conf/podc/PadalkinS24},
convexity \cite{DBLP:journals/siamcomp/Keil85},
and
geodesic convexity \cite{DBLP:journals/tcs/CoyCSSW24}.
Special cases of the convex decompositions include the well known trapezoidal decomposition \cite{DBLP:journals/comgeo/Seidel91} and triangulation \cite{DBLP:journals/dcg/Chazelle91}.

In this paper, we investigate the computation of simple, geodesically convex decompositions in the context of programmable matter, i.e., a substance that can adaptively change its physical properties \cite{DBLP:journals/ijhsc/ToffoliM93}.
We consider systems consisting of many small, primitive particles that can communicate with each other to perform tasks on a global level.
In the past, many exciting applications for programmable matter have been proposed, including self-healing materials \cite{an2021self} and minimally invasive surgery \cite{montemagno1999constructing}.
The \textit{geometric amoebot model} \cite{DBLP:conf/spaa/DerakhshandehDGRSS14} is an established model to study active programmable matter systems.
In this model, the particles, called \textit{amoebots}, are placed on an infinite triangular grid graph, where only neighboring amoebots can directly communicate with each other.
Since information can only travel amoebot by amoebot, most problems have a natural lower bound of $\Omega(\mathit{diam})$ where $\mathit{diam}$ denotes the diameter of the structure.

For that reason, we employ the \textit{reconfigurable circuit extension} \cite{DBLP:journals/jcb/FeldmannPSD22}, which allows the instantaneous transmission of very simple signals to far-away amoebots along so-called \textit{circuits}.
Past work on this extension showed that it can enable (poly)logarithmic solutions for many problems \cite{DBLP:conf/wdag/ArtmannPS25,DBLP:journals/jcb/FeldmannPSD22,DBLP:conf/podc/PadalkinS24,DBLP:journals/nc/PadalkinSW24}.
Feldmann \emph{et al.} \cite{DBLP:journals/jcb/FeldmannPSD22} proposed a linear-time triangulation algorithm for a specific class of amoebot structures, and Padalkin and Scheideler \cite{DBLP:conf/podc/PadalkinS24} showed how to decompose a simple structure into so-called tunnel regions (formally defined below) in logarithmic time.
Note that the former algorithm suffers from a high runtime.
Further, both decompositions have restrictions on the amoebot structure.
In this paper, we define a decomposition into $O(|\mathcal H|)$ simple, geodesically convex regions for arbitrary amoebot structures with $|\mathcal{H}|$ holes, and show how it can compute such a decomposition in $O(\log n)$ rounds.

\subsection{Problem Statement and Our Contribution}
\label{sec:problem_statement}

Let $G_\Delta = (V_\Delta, E_\Delta)$ be the infinite regular triangular grid and let $\Gamma = (V, E)$ be an arbitrary connected subgraph of $G_\Delta$.
Each node of $\Gamma$ represents an amoebot and each edge indicates neighboring amoebots.
Let $R \subseteq V$.
We call the maximally connected components of $G_{V_\Delta \setminus R}$ holes, where $G_{V_\Delta \setminus R} = G_\Delta|_{V_\Delta \setminus R}$ is the graph induced by $V_\Delta \setminus R$.
There is a unique unbounded hole, which we refer to as the outer hole.
All other holes of $R$ are its inner holes and we denote the set of inner holes of $R$ with $\mathcal{H}_R$.
We call $R$ a \emph{region} if and only if it is connected in $\Gamma$.
We call $R$ \emph{simple} if and only if $\mathcal{H}_R = \emptyset$, i.e., $R$ contains no inner holes.
We call $R$ \emph{geodesically convex}, or just \emph{convex}, if and only if for each $u,v \in R$, every shortest path between $u$ and $v$ in $\Gamma$ is completely contained within $R$.

A \emph{decomposition} is a family $\mathcal R$ of regions $R \subseteq V$ such that $V = \bigcup_{R \in \mathcal R} R$.
Note that we do not require the regions to be disjoint.
A decomposition is \emph{geodesically convex}, or just \emph{convex}, if and only if each region $R \in \mathcal R$ is simple and convex.

We consider the \emph{convex decomposition problem}.
We say that the amoebot structure computes a decomposition $\mathcal R$ if and only if for each region $R \in \mathcal R$, each amoebot $u \in R$ knows $R \cap N(u)$, where $N(u)$ denotes the neighborhood of $u$ in $\Gamma$.
The goal of the amoebot structure is to compute a convex decomposition.
The main result of this paper is stated in the following theorem.

\begin{theorem}
\label{th:main_theorem}
    An amoebot structure with $|\mathcal{H}|$ holes computes a decomposition consisting of $\Theta(|\mathcal{H}|)$ simple geodesic convex regions within $O(\log n)$ rounds, w.h.p.\footnote{An event holds \emph{with high probability (w.h.p.)} if it holds with probability at least $1 - 1/n^c$ where the constant $c$ can be made arbitrarily large.}
\end{theorem}

As a byproduct, we improve the global maxima algorithm of \cite{DBLP:journals/nc/PadalkinSW24} for special cases and with that also their spanning tree algorithm to $O(\log n)$ rounds w.h.p.


\subsection{Related Work}

In the computational geometry community, decompositions have been extensively studied:
Given an $n$-gon with $r$ reflex vertices, the task is to decompose it into a minimum number of ``simpler'' subpolygons.
The complexity of the problem depends on 
(i) whether the polygon has holes,
(ii) whether Steiner points are allowed, and
(iii) which shapes the subpolygons are allowed to have.
Shapes can be, for example, convex \cite{chazelle1985optimal}, triangular \cite{DBLP:journals/dcg/AmatoGR01,bern1995mesh,DBLP:journals/dcg/Chazelle91}, geodesic triangular \cite{DBLP:journals/algorithmica/ChazelleEGGHSS94,DBLP:journals/dcg/SpeckmannT05,DBLP:journals/dcg/Streinu05}, trapezoidal \cite{DBLP:journals/jacm/AsanoAI86,DBLP:journals/comgeo/Seidel91}, star-convex \cite{DBLP:journals/theoretics/AbrahamsenBNZ26}, or monotone \cite{DBLP:journals/ipl/LiuN88}.


For simple polygons, the best known convex decomposition algorithm with Steiner points runs in time $O(n^3)$ \cite{chazelle1985optimal}, and without Steiner points in time $O(n + r^2 \min \{ r^2, n \})$ \cite{DBLP:journals/ijcga/KeilS02}.
For the latter case, Chazelle proposed a $\frac{13}{3}$-approximation algorithm that runs in time $O(n \log n)$ \cite{DBLP:conf/focs/Chazelle82}.
However, for general polygons, finding a convex decomposition with a minimum number of convex components is NP-complete \cite{DBLP:journals/tit/ORourkeS83}.
For that reason, approximate convex decompositions were introduced, where some non-convex features are considered less significant and are consequentially ignored, e.g., \cite{DBLP:journals/comgeo/LienA06,DBLP:conf/icip/MamouG09,DBLP:journals/tog/WeiLLS22}.


For polygons with $|\mathcal{H}|$ holes, it is possible to compute a triangulation in time $O(n + |\mathcal{H}|\log |\mathcal{H}|)$ \cite{DBLP:journals/corr/abs-2603-21617}.
There is also a randomized algorithm solving the problem in expected time $O(n)$, which is considerably simpler \cite{DBLP:journals/dcg/AmatoGR01}.
This matches the known lower bound of $\Omega(n \log n)$ for general polygons \cite{DBLP:books/lib/BergCKO08}.


A geodesic triangle (also called pseudo-triangle) is a polygon with exactly three convex corners.
Note that any triangulation is also a geodesic triangulation (also called pseudo-triangulation).
A geodesic triangulation is called pointed if and only if each vertex has an incident face with an angle larger than $180^\circ$.
It turns out that the definition of a minimum geodesic triangulation is equivalent to the one of a pointed geodesic triangulation \cite{DBLP:journals/dcg/Streinu05}.
For simple polygons, a pointed geodesic triangulation can be computed in linear time \cite{DBLP:journals/algorithmica/ChazelleEGGHSS94,DBLP:journals/dcg/SpeckmannT05,DBLP:journals/dcg/Streinu05}, and for general polygons, in time $O(n \log n)$ \cite{DBLP:journals/dcg/Streinu05}.


The literature on graph decompositions is vast.
One of the most important classes of graph decompositions is the tree decomposition (first introduced in \cite{DBLP:journals/jal/RobertsonS86}, see, e.g., \cite{DBLP:journals/iandc/BodlaenderK10,DBLP:journals/iandc/BodlaenderK11} for surveys), where the nodes of a graph are covered with subsets called bags, that are associated with the nodes of a tree $T$. 
The decompositions ensure that every node is in at least one bag and, for every edge, at least one bag contains both of its endpoints.
Additionally, if a node appears in two bags, it must also appear in every bag corresponding to a tree node on the unique path in $T$ between their corresponding tree nodes.

Another important class of graph decompositions is the core decomposition (first introduced in \cite{seidman1983network}, see, e.g., \cite{DBLP:journals/vldb/MalliarosGPV20} for a recent survey).
For $k\in\mathbb{N}$, the $k$-core of a graph is the largest subgraph where every node has degree at least $k$. 
It can be constructed by iteratively removing nodes of degree smaller than $k$ from the graph until it ceases to change.
By construction, the $(k+1)$-core of any graph is a subset of its $k$-core, i.e., the resulting cores form a hierarchy on the nodes based on their connectivity.

The decomposition we present in this paper is closely related to a decomposition presented for square grid graphs in \cite{DBLP:journals/tcs/CoyCSSW24}.
The authors motivate their decomposition to enable routing on a square grid graph with holes.
To this end, they ensure that each of the $O(|\mathcal{H}|)$ node sets created by the decomposition is hole-free and what they call path-convex, i.e., for each pair of nodes in the set, at least one shortest path remains entirely inside the region.
They prove that traversing the right sequence of regions yields a shortest path in square grid graphs.
Their work employs a construction of \cite{DBLP:journals/corr/abs-2202-08008} to extend their results to Unit Disk Graphs where nodes have arbitrary positions in $\mathbb{R}^2$ and are connected if and only if they have euclidean distance $1$.
Contrary to their approach, we focus on triangular grid graphs and show that for each node pair of the same region \emph{all} shortest paths stay within that region.
Moreover, we present implementation details for the amoebot model.


The reconfigurable circuit extension to the geometric amoebot model was introduced by Feldmann \emph{et al.} \cite{DBLP:journals/jcb/FeldmannPSD22}.
Since then, polylogarithmic solutions were proposed for various problems including leader election \cite{DBLP:journals/jcb/FeldmannPSD22}, orientation agreement \cite{DBLP:journals/jcb/FeldmannPSD22}, shape recognition \cite{DBLP:journals/jcb/FeldmannPSD22,DBLP:journals/nc/PadalkinSW24}, spanning trees \cite{DBLP:conf/wdag/EmekGH24,DBLP:journals/nc/PadalkinSW24}, shape containment \cite{DBLP:conf/wdag/ArtmannPS25}, and shortest path forests \cite{DBLP:conf/podc/PadalkinS24}.
The universal shape recognition algorithm by Feldmann \emph{et al.} \cite{DBLP:journals/jcb/FeldmannPSD22} is based on a triangulation.
The runtime of the triangulation depends on the shape of the amoebot structure and is linear in $n$ in the worst case.
In order to avoid a linear runtime, they abort the triangulation of the amoebot structure if it takes longer than the triangulation of the input shape.
Padalkin and Scheideler \cite{DBLP:conf/podc/PadalkinS24} utilize a divide and conquer approach to construct shortest path forests within $O(|\mathcal H| \log^3 n)$ rounds.
For that, they decompose a simple amoebot structure into tunnel regions.
We will use their decomposition algorithm in our second phase (see \Cref{sec:tunnel_regions}).
The shortest path forests decompose an amoebot structure with respect to a set of sources.
However, in general, this decomposition is neither simple nor geodesically convex.


\section{Preliminaries}\label{sec:Triangular_grid_graphs}

For a triangular grid graph $\Gamma=(V_\Gamma,E_\Gamma)$, we define a path $\Pi = (v_0, \dots, v_k)$ as a sequence of nodes with $\{v_i, v_{i+1}\} \in E_\Gamma$ for all $i \in \{0, \dots, k-1\}$.
Slightly abusing notation, we write $v \in \Pi$ if $v \in \{v_0, \dots, v_k\}$ and denote the length of the path by $|\Pi|$ which is the number of edges of $\Pi$. 
If $v_0 = v_k$, $\Pi$ is called a cycle. 
If $u = v_0$ and $v = v_k$, $\Pi$ is called a $uv$-path.
Further, if there is an $i\in \{1,\dots k-1\}$ s.t. $w=v_i$, $\Pi$ is called a $uwv$-path.
For a $uv$-path $\Pi$ and a $vw$-path $\Pi'$, we denote by $\Pi \circ \Pi'$ the path obtained by concatenating $\Pi$ and $\Pi'$.

The distance between two nodes $u,v \in V_\Gamma$ of some triangular grid graph $\Gamma = (V_\Gamma, E_\Gamma)$ is the number of edges on a shortest $uv$-path, i.e. $d_\Gamma(u,v) \coloneqq \min_{\text{$uv$-path $\Pi$}}|\Pi|$.
We extend the distance function to sets by taking the minimum distance between two nodes.
More specifically, we define for $V_1,V_2\subseteq V_\Gamma$ distance $d_\Gamma(V_1,V_2)\coloneqq \min_{v_1\in V_1,\:v_2\in V_2}d_\Gamma(v_1,v_2)$.

We define the boundary of an (inner or outer) hole $H$ to be the set of nodes that neighbor a node of $H$ and call a node incident to $H$, if it is part of $H$'s boundary.
A node set is incident to $H$, if one of its nodes is incident to $H$.

As in previous works on decompositions of grid graphs, our construction is based on the concept of \emph{portal graphs} \cite{DBLP:journals/tcs/CoyCSSW24,DBLP:conf/podc/PadalkinS24}.
For a triangular grid graph $\Gamma = (V_\Gamma, E_\Gamma)$, let $E_x \subseteq E_\Gamma$ be the set of edges parallel to the $x$-axis.
Then, the \emph{$x$-portals} of $\Gamma$ are the connected components of the subgraph $(V_\Gamma, E_x)$.
For each $u \in V$, we denote the unique $x$-portal that contains $u$ by $\portal_x(u)$.
We define \emph{$y$-} and \emph{$z$-portals} analogously.
Two portals $P_1$ and $P_2$ are adjacent if there exists an edge $\{v_1,v_2\} \in E_\Gamma$ such that $v_1 \in P_1$ and $v_2 \in P_2$.

The \emph{$x$-portal graph} $\mathcal P_x = (V_{\mathcal{P}_x}, E_{\mathcal{P}_x})$ of $\Gamma$ has a node for each $x$-portal of $\Gamma$.
Two nodes of $\mathcal P_x$ are connected by an edge if and only if the corresponding portals are adjacent.
We define the $y$-portal graph $\mathcal P_y$ and the $z$-portal graph $\mathcal P_z$ analogously.
An example of the three portal graphs is depicted in \Cref{fig:portal_graphs}.

\begin{figure}[tbp]
    \centering
    \includegraphics[scale=.6]{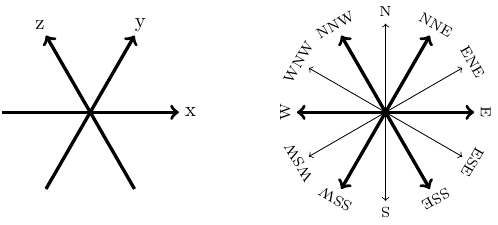} 
    
    \vspace{2mm}
    
    \includegraphics[scale=.9]{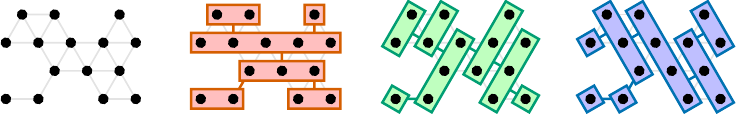}    
    \caption{
        A triangular grid graph and the corresponding portal graphs $\mathcal{P}_x$, $\mathcal{P}_y$ and $\mathcal{P}_z$ (from left to right).
        The figure is adapted from \cite{DBLP:conf/podc/PadalkinS24}.
    }
    \label{fig:portal_graphs}
\end{figure}

The definition of portal graphs yields a notion of distances in the direction of a specific axis. 
For two nodes $u,v \in V_\Gamma$, we define the $x$-distance $d_{\Gamma,x}(u,v)$ between $u$ and $v$ as the distance between $\portal_x(u)$ and $\portal_x(v)$ in $\mathcal{P}_x$.
Similarly, we define $d_{\Gamma,y}(u,v)$ and $d_{\Gamma,z}(u,v)$ as the distances in $\mathcal{P}_y$ and $\mathcal{P}_z$, respectively.

We conclude by restating two important lemmas from \cite{DBLP:journals/tcs/CoyCSSW24,DBLP:conf/podc/PadalkinS24}.

\begin{lemma}[Adapted from \cite{DBLP:journals/tcs/CoyCSSW24,DBLP:conf/podc/PadalkinS24}]
\label{lem:portal_tree}
    Let $\Gamma = (V_\Gamma, E_\Gamma)$ be a simple triangular grid graph.
    Then, the portal graphs $\mathcal{P}_x$, $\mathcal{P}_y$, and $\mathcal{P}_z$ are trees.
\end{lemma}

\begin{lemma}[Adapted from \cite{DBLP:journals/tcs/CoyCSSW24,DBLP:conf/podc/PadalkinS24}]
\label{lem:portal_distances}
    Let $\Gamma = (V_\Gamma, E_\Gamma)$ be a simple triangular grid graph. 
    Then, $d_\Gamma(u,v) = \frac{1}{2}(d_{\Gamma,x}(u,v) + d_{\Gamma,y}(u,v) + d_{\Gamma,z}(u,v))$ holds.
\end{lemma}

\section{Convex Decomposition}
\label{sec:pathconvex_region_decomposition}


In this section, we present our convex decomposition for regular triangular grid graphs.
Our construction closely follows the approach of Coy \emph{et al.} \cite{DBLP:journals/tcs/CoyCSSW24} for square grid graphs.
However, since we consider triangular grid graphs, the proofs turn out to be more complicated.
The general idea of the approach is to split the grid graph at certain strategic portals and nodes.
All nodes involved in a split become part of all adjacent resulting regions.

Our construction consists of three phases.
In the first phase, we decompose the grid graph into simple regions.
We call the intersection of a region and a portal at which we split the grid graph a \emph{gate}.
In the second phase, we split the regions further into \emph{tunnel regions}, i.e., regions that intersect at most two gates.
In the third phase, we divide the tunnel regions into convex regions.

Throughout the paper, we frequently split triangular grid graphs along portals.
In the following, we describe how we perform such splitting operations for $y$-portals.
Splitting operations for $x$- and $z$- portals are defined analogously.
An example of a splitting operation is depicted in \Cref{fig:splitting_operations}.

\begin{definition}[Splitting Operations]
    \label{def:splitting_operations}
    Let $\Gamma=(V_\Gamma,E_\Gamma)$ be a triangular grid graph and $P$ be a $y$-portal in $\Gamma$.
    Further, let $v_1, \dots, v_\ell$ be nodes on $P$ that are each adjacent to a distinct specified empty grid point.
    We distinguish two types of splitting operations.
    \begin{enumerate}
        \item \textbf{Splitting $\Gamma$ at $P$:}
        We replace each node $p\in P$ with two copies $p_{\mathit{WNW}}$ and $p_{ESE}$.
        Let $q\in P$ and $r\in P$ be the neighbors of $p$ in $P$ in $\mathit{NNE}$ and $\mathit{SSW}$ direction, if they exist, respectively.
        $p_\mathit{WNW}$ has edges to $q_\mathit{WNW}$, $r_\mathit{WNW}$, and to $p$'s neighbors in the $\mathit{NNW}$ and $W$ directions, if they exist, respectively.
        Analogously, $p_\mathit{ESE}$ has edges to $q_\mathit{ESE}$, $r_\mathit{ESE}$, and to $p$'s neighbors in the $\mathit{SSE}$ and $E$ directions, if they exist, respectively.
        \item \textbf{Splitting $\Gamma$ at $P$ and $v_1,\dots,v_\ell$:}
        We start by splitting $\Gamma$ at $P$ as described in Case (1).
        Afterwards, for $1\leq i\leq \ell$, if the empty gridpoint specified for $v_i$ is in the directions $SSE$ or $E$ from $v_i$, we replace $v_{\mathit{ESE}}$ with two copies $v_{\mathit{ESE}}^{\mathit{NNE}}$ and $v_{\mathit{ESE}}^{\mathit{SSW}}$.
        $v_{\mathit{ESE}}^{\mathit{NNE}}$ has edges to the neighbors of $v_\mathit{ESE}$ in the $\mathit{NNE}$ and $\mathit{E}$ directions, if they exist, respectively.
        Analogously, $v_{\mathit{ESE}}^{\mathit{SSW}}$ has edges to the neighbors of $v_\mathit{ESE}$ in the $\mathit{SSW}$ and $\mathit{SSE}$ directions, if they exist, respectively.
        The case where the empty grid point specified for $v_i$ is in the directions $\mathit{NNW}$ or $W$ from $v_i$ is analogous.
    \end{enumerate}
\end{definition}

\begin{figure}[tb]
    \begin{minipage}[t]{.3\textwidth}
        \centering
        \includegraphics[scale=1.3]{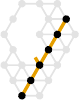}
        \subcaption{A grid graph with a splitting portal and a splitting node.}
    \end{minipage}
    \hfill
    \begin{minipage}[t]{.3\textwidth}
        \centering
        \includegraphics[scale=1.3]{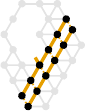}
        \subcaption{The grid graph has been cut along the splitting portal.}
    \end{minipage}
    \hfill
    \begin{minipage}[t]{.3\textwidth}
        \centering
        \includegraphics[scale=1.3]{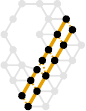}
        \subcaption{The grid graph has been cut along the splitting portal and at the splitting node.}
    \end{minipage}
    \caption{A detailed depiction of the splitting operations (see \Cref{def:splitting_operations}). The nodes connected by a dotted line occupy the same grid point.}
    \label{fig:splitting_operations}
\end{figure}

\subsection{Decomposition into Simple Regions}
\label{sec:simple_regions}

The first phase splits our grid graph into simple regions, i.e., regions without holes.
For each inner hole $H \in \mathcal{H}$ of $\Gamma$ we consider the $\mathit{WNW}$-most node $v_{\mathit{WNW}}(H)$ and the $\mathit{ESE}$-most node $v_{\mathit{ESE}}(H)$ on the boundary of $H$ (breaking ties by picking the $\mathit{NNE}$-most node respectively).
We conduct splits at $\portal_y(v_{\mathit{WNW}}(H))$ and $v_{\mathit{WNW}}(H)$ w.r.t.\ $H$, and $\portal_y(v_{\mathit{ESE}}(H))$ and $v_{\mathit{ESE}}(H)$ w.r.t.\ $H$ (see \Cref{def:splitting_operations} Case (2)).
An example of the procedure is depicted in \Cref{fig:splitting_operations:a,fig:splitting_operations:b}.

\begin{figure}[tb]
    \begin{minipage}[t]{.3\textwidth}
        \centering
        \includegraphics[width=0.9\textwidth]{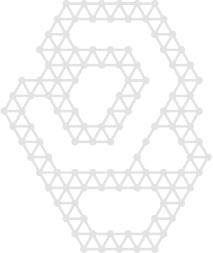}
        \subcaption{A grid graph with four holes.}
        \label{fig:splitting_operations:a}
    \end{minipage}
    \hfill
    \begin{minipage}[t]{.3\textwidth}
        \centering
        \includegraphics[width=0.9\textwidth]{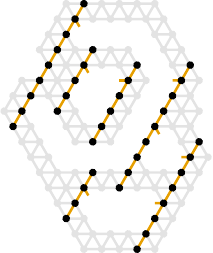}
        \subcaption{The result of the simple region decomposition.}
        \label{fig:splitting_operations:b}
    \end{minipage}
    \hfill
    \begin{minipage}[t]{.3\textwidth}
        \centering
        \includegraphics[width=0.9\textwidth]{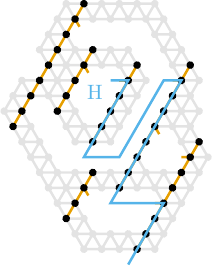}
        \subcaption{The blue line corresponds to $S_\searrow$ for hole $H$.}
        \label{fig:splitting_operations:c}
    \end{minipage}
    \caption{A detailed depiction of the simple region decomposition and one of the sequences $S_\searrow$ used in the proof of \Cref{lem:decomposition:simple}. No path can cross $S_\searrow$.}
    \label{fig:simple_decomposition}
\end{figure}

We prove that the construction only creates $O(|\mathcal{H}|)$ regions and that the created regions are simple triangular grid graphs.

\begin{lemma}
\label{lem:decomposition:simple}
We decompose a triangular grid graph into at most $3|\mathcal{H}|+1$ simple regions. Each resulting region is a triangular grid graph.
\end{lemma}
\begin{proof}
    Each split along a portal increases the number of regions by at most one. Splitting at the $\mathit{WNW}$-most and the $\mathit{ESE}$-most node further increases the number of regions by at most one.
    Thus, we increase the number of regions by at most three for each hole (two from portals and one from the splitting nodes).
    Consequently, the total number of regions will not exceed $3|\mathcal{H}|+1$.
    
    To prove that the regions resulting from the procedure are simple, we argue that the splitting actions induce four sequences $S_\searrow$, $S_\nearrow$, $S_\nwarrow$, and $S_\swarrow$ for each hole $H\in \mathcal{H}$.
    Each sequence starts with $H$, followed alternatingly by splitting portals and holes until it ends with the outer hole. 
    We describe the construction of $S_\searrow$ in detail and provide a depiction of $S_\searrow$ in \Cref{fig:splitting_operations:c}.
    The other sequences are defined analogously.
    We start by adding $H$ to $S_\searrow$.
    After adding an inner hole $H'$ to $S_\searrow$, we add the splitting portal $\portal_y(v_{\mathit{ESE}}(H'))$ to $S_\searrow$.
    After adding a splitting portal $P$ to $S_\searrow$, we add the hole that is in the $\mathit{SSW}$ direction of $P$'s $\mathit{SSW}$-endpoint to $S_\searrow$.
    By construction, we never add a hole twice.
    Thus, we eventually add the outer hole to $S_\searrow$ and the construction terminates.
    By construction, no region boundary can pass any of the sequences $S_\searrow$, $S_\nearrow$, $S_\nwarrow$, or $S_\swarrow$.
    Thus, no region can enclose $H$.
    As we can perform this argument for every hole, no region can enclose any hole, i.e. the regions must be simple.
    
    It remains to prove that the resulting regions are triangular grid graphs.
    This can only be false if two copies of the same node end up in the same region.
    As no path inside a region can cross any of the above sequences, however, each node copy must end up in a different region.
\hfill $\square$ 
\end{proof}

\begin{lemma}
\label{lem:number_of_gates}
    The construction creates at most $6|\mathcal{H}|$ gates.
\end{lemma}
\begin{proof}
    For each hole $H\in\mathcal{H}$ we split at $\portal_y(v_{\mathit{WNW}}(H))$ and $v_\mathit{WNW}(H)$ as well as $\portal_y(v_{\mathit{ESE}}(H))$ and $v_\mathit{ESE}(H)$. For $\portal_y(v_{\mathit{WNW}}(H))$, this creates at most three gates --- one in the $\mathit{WNW}$ direction of $\portal_y(v_{\mathit{WNW}}(H))$, one in the $\mathit{ESE}$ direction of $\portal_y(v_{\mathit{WNW}}(H))$ and in the $\mathit{NNE}$ direction of $v_{\mathit{WNW}}(H)$,  one in the $\mathit{ESE}$ direction of of $\portal_y(v_{\mathit{WNW}}(H))$ and in the $\mathit{SSW}$ direction of $v_{\mathit{WNW}}(H)$. For $\portal_y(v_{\mathit{ESE}}(H))$, we can argue analogously, that at most three gates are created. Note that any overlapping gates can only decrease the total number of gates.
\hfill $\square$ 
\end{proof}

We conclude the section by stating an observation that we frequently use implicitly throughout the paper.

\begin{observation}
    Let $R$ be a simple region and let $P$ be a portal through $R$. 
    If we split $R$ at $P$ (and possibly some nodes on $P$), then the resulting regions remain simple and no region contains multiple copies of the same node.
\end{observation}


\subsection{Decomposition into Tunnel Regions}
\label{sec:tunnel_regions}

Next, we decompose each simple region into tunnel regions, i.e., regions that intersect at most two gates.
We can adopt the construction used by Coy \emph{et al.} \cite{DBLP:journals/tcs/CoyCSSW24} for square grids without any changes
since it only relies on portal graphs.
For the sake of completeness, we still describe the construction in the following.
However, we deviate from Coy \emph{et al.} and use the description by Padalkin and Scheideler \cite{DBLP:conf/podc/PadalkinS24} since we will use the techniques of the latter to compute the tunnel regions in \Cref{sec:amoebot:decomposition}.
See \Cref{fig:tunnel_decomposition} for an example.

Consider an arbitrary region and its portal graph $\mathcal P_y$.
Since the region is simple, $\mathcal P_y$ is a tree.
First, we iteratively prune leaves that are not gates until each leaf is a gate.
Let $\mathcal P'_y$ be the resulting portal tree.
Then, we split the region at each non-gate portal in $\mathcal P'_y$ of degree at least $3$.

Consider one of the resulting regions and its portal tree $\mathcal P''_y$ without the previously pruned portals.
By construction, each leaf and each portal of degree at least $3$ in $\mathcal P''_y$ is a gate. 
We split each gate $G$ of degree at least $2$ as follows.
Observe that by definition of gates, all adjacent portals of $G$ are either on the west or east side of $G$.
This allows us to order the adjacent portals $P_1, \dots, P_\ell$ from north to south.
Let $g_i$ denote the northernmost node of $G$ adjacent to a node of $P_i$.
We split the region at nodes $g_2, \dots, g_\ell$.
This splits $G$ into $\ell$ gates of degree $1$ in the portal tree.
Since now each gate has a degree of $1$, we obtain the following lemma.

\begin{figure}[tbp]
    \centering
    \includegraphics[scale=1.2]{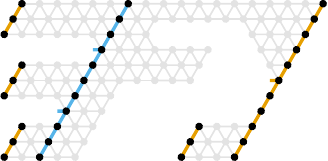}    
    \caption{
        Example of splitting a simple region into tunnel regions. Initially, the region intersects the $5$ orange gates. After performing a split along the blue gate and a node split at the east-most orange gate, the region is decomposed into $5$ tunnel regions.
    }
    \label{fig:tunnel_decomposition}
\end{figure}

\begin{lemma}[Adapted from \cite{DBLP:conf/podc/PadalkinS24}]
\label{lem:tunnel_region}
    We decompose a simple region that intersects $k$ gates into $\Theta(k)$ tunnel regions.
\end{lemma}

Since by \Cref{lem:number_of_gates}, there are $\Theta(|\mathcal H|)$ gates overall, we obtain the following result.

\begin{corollary}
\label{cor:tunnel_regions}
    We decompose the grid graph into $\Theta(|\mathcal H|)$ tunnel regions.
\end{corollary}


\subsection{Decomposition into Convex Regions}
\label{sec:pathconvex_regions}

In the final phase, we decompose each tunnel region into convex regions.
Since the tunnel regions can take on complex shapes, there can still be pairs of nodes in the same tunnel region, such that the shortest path in the entire grid graph $\Gamma$ between them \textit{leaves} the tunnel and later \textit{reenters} it from the other gate.
The goal of this phase will be to separate those pairs of nodes by additional splits along portals.
In particular, we will see that for each tunnel only a constant number of additional splits will be needed to make the tunnels convex.

Consider a tunnel region $T$ with gates $G$ and $G'$.
Remember that $G$ and $G'$ are portals in the $y$-direction.
We first split $T$ by portals in $x$- and $z$-direction that intersect $G$ or $G'$.
For both directions $q \in \{x,z\}$, we distinguish between the following two cases (see \Cref{fig:convex_decomposition:a,fig:convex_decomposition:b,fig:convex_decomposition:d} for examples):

\begin{enumerate}
    \item If there is a $q$-portal that intersects both $G$ and $G'$ within $T$, let $P_{\uparrow,q}$ and $P_{\downarrow,q}$ be the northernmost and southernmost of these $q$-portals, respectively.
    Then, we split $T$ at $P_{\uparrow,q}$ and $P_{\downarrow,q}$.
    Furthermore, in case $P_{\uparrow,q}$ touches a boundary node $b_{\uparrow,q} \not\in G \cup G'$ to the north of $P_{\uparrow,q}$, we split the region to the north of $P_{\uparrow,q}$ at the closest such node to $G$. 
    We do this symmetrically for the region below $P_{\downarrow,q}$.
    \item Otherwise, let \gate{G}{q} be the unique $q$-portal intersecting $G$ that minimizes $d_{T,q}(\gate{G}{q},G')$.
    Analogously, let \gate{G'}{q} be the unique $q$-portal intersecting $G'$ that minimizes $d_{T,q}(\gate{G'}{q},G)$.
    Then, we split $T$ at $\gate{G}{q}$ and $\gate{G'}{q}$.
    Note that after splitting $T$ at (only) $\gate{G}{q}$ and $\gate{G'}{q}$ there is a region $R_q$ enclosed by $\gate{G}{q}$ and $\gate{G'}{q}$.
    In case $\gate{G}{q}$ touches a boundary node $\boundarynode{G}{q} \not\in G$ from the side of $R_q$, we split $R_q$ at the closest such node to $G$.
    Similarly, in case $\gate{G'}{q}$ touches a boundary node $\boundarynode{G'}{q} \not\in G'$ from the side of $R_q$, we split $R_q$ at the closest such node to $G'$.
\end{enumerate}

\begin{figure}[tb]
    \centering
    \begin{minipage}[t]{.45\textwidth}
        \centering
        \setcounter{subfigure}{0}
        \includegraphics[width=0.85\textwidth]{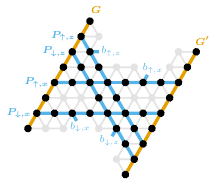}
        \subcaption{Case 1 for both $q=x$ and $q=z$}
        \label{fig:convex_decomposition:a}
    
        \vspace{0.5cm}

        \setcounter{subfigure}{2}
        \includegraphics[width=0.85\textwidth]{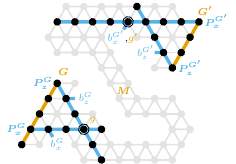}
        \subcaption{Case 2 for both $q=x$ and $q=z$}
        \label{fig:convex_decomposition:d}
    \end{minipage}
    \hspace{0.03\textwidth}
    \begin{minipage}[t]{.45\textwidth}
        \centering
        \setcounter{subfigure}{1}
        \includegraphics[width=0.85\textwidth]{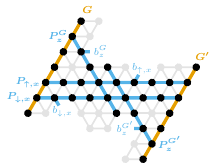}
        \subcaption{Case 1 for $q=x$ and case 2 for $q=z$}
        \label{fig:convex_decomposition:b}
    
        \vspace{0.35cm}
    
        \setcounter{subfigure}{3}
        \includegraphics[width=\textwidth]{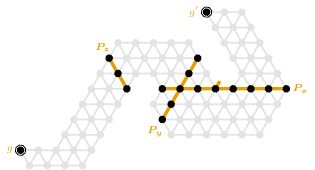}
        \subcaption{Splitting a region $M$ with gates consisting of single nodes.}
        \label{fig:convex_decomposition:point_shaped}
    \end{minipage}
    \caption{Examples of splitting tunnel regions into convex regions. Black circles around nodes indicate the nodes $g$ and $g'$ which are interpreted as single node gates for splitting according to Definition \ref{def:point_shape_decomposition}.}
    \label{fig:convex_decomposition}
\end{figure}

Note that if for both values of $q$ the second case applies (see \Cref{fig:convex_decomposition:d}), there is a unique region $M$ that intersects at least one gate on $G$'s side ($\gate{G}{x}$ or $\gate{G}{z}$) and at least one gate on $G'$'s side ($\gate{G'}{x}$ or $\gate{G'}{z}$).
We now argue that all resulting regions, except $M$, are convex.
To this end, we make the following two observations.

\begin{observation}
\label{obs:simple_case}
    Let $R$ be a region of some grid graph $\Gamma$ that intersects (a) exactly one gate or (b) intersects exactly two gates that meet in an obtuse angle (i.e. the angle between the gates within $R$ is $120^\circ$).
    Then $R$ is convex, as any path leaving and reentering $R$ can be made shorter by traversing the gate(s) instead.
\end{observation}

\begin{lemma}
\label{lem:portal_splitting_remain_pathconvex}
    Let $R$ be a convex region and let $P$ be a portal through $R$. 
    If we split $R$ at $P$ (and possibly some nodes on $P$), then the resulting regions remain convex.
\end{lemma}
\begin{proof}
    Let $R', R'' \subseteq R$ denote the regions resulting from splitting $R$ at $P$.
    W.l.o.g., let $u,v \in R'$.
    Since $R$ is convex, every shortest $uv$-path stays inside $R$.
    Now, assume that there is a shortest $uv$-path $\Pi$ that crosses $P$ and reenters $R'$ later through $P$.
    Then, this contradicts the fact that $\Pi$ is a shortest path, as $\Pi$ can be shortened by traversing $P$ instead.
    Thus, a shortest path between $u$ and $v$ stays within $R'$.
    As this is true for all nodes of $R'$, $R'$ is convex.
\hfill $\square$ 
\end{proof}

Intuitively, our construction ensures that all resulting regions, except $M$, are (subregions of regions that are) bordered by exactly to gates that meet in an obtuse angle, parallelograms or triangles, which are all convex.
This gives the following lemma.

\begin{lemma}
\label{lem:path_convex_regions_correctness}
    Except for region $M$, all resulting regions are convex.
\end{lemma}
\begin{proof}
    First, assume that for at least one value of $q$ the first case applies.
    W.l.o.g. assume that the first case applies for $q=x$ (the other case is analogous).
    Then, after splitting the tunnel at $P_{\downarrow,x}$ and $P_{\uparrow,x}$ in $x$-direction, we end up with up to 5 different regions.
    The region between $P_{\downarrow,x}$ and $P_{\uparrow,x}$ is a parallelogram and thus trivially convex.
    Furthermore, the region defined by gates $G$ and $P_{\downarrow,x}$ below $P_{\downarrow,x}$ and the region defined by gates $G'$ and $P_{\uparrow,x}$ above $P_{\uparrow,x}$ are convex by Observation \ref{obs:simple_case}, as the gates intersect in an obtuse angle.
    
    Now consider the two remaining regions, that are (so far) defined by two gates that intersect in an acute angle.
    Both of these regions are intersected by a gate in $z$-direction, no matter whether $q=z$ falls into case 1 or 2.
    To see why this is the case, consider one of the remaining regions, for example the region $R$ that is defined by $G$ and $P_{\uparrow,x}$ above $P_{\uparrow,x}$.
    If $q=z$ falls into case 1, then $P_{\uparrow,z}$ touches a boundary node on its north-eastern side by definition, because the region between $P_{\uparrow,z}$ and $P_{\downarrow,z}$ encompasses \textit{all} nodes with distance 0 between $G$ and $G'$ in $z$-direction.
    Similarly, if $q=z$ falls into case 2, then $\gate{G}{z}$ touches a boundary node on its north-eastern side because it minimizes the distance in the $z$-portal graph from $G$ to $G'$ by definition.
    Now note that in any case, one of these boundary nodes has to be inside $R$, because otherwise $P_{\uparrow,x}$ wouldn't connect $G$ with $G'$ since $T$ is simple.
    Thus, a gate in $z$-direction indeed splits up $R$.
    More precisely, it splits the region into up to three parts, depending on where it touches its first boundary node.
    The same argument can be made to show that the region defined by $G'$ and $P_{\downarrow,x}$ below $P_{\downarrow,x}$ is split up by a gate in $z$-direction into up to three parts.
    One of these parts is a triangle and thus trivially convex.
    The other two parts are defined by two gates that intersect in an obtuse angle and are thus convex by Observation \ref{obs:simple_case}.
    Finally, note that adding gates along portals to already convex regions preserves convexity by \Cref{lem:portal_splitting_remain_pathconvex}.
    Therefore, all regions created are convex.

    The last case to consider is that for both $q=x$ and $q=z$ the second case applies.
    Similarly to before, after splitting at $\gate{G}{x}$ and $\gate{G'}{x}$ in the first step, up to 5 regions are created.
    The region defined by gates $G$ and $\gate{G}{x}$ below $\gate{G}{x}$ and the region defined by gates $G'$ and $\gate{G'}{x}$ above $\gate{G'}{x}$ are convex by Observation \ref{obs:simple_case}, as the gates intersect in an obtuse angle.
    The remaining two regions (besides $M$) are defined by two gates that intersect in an acute angle.
    Denote by $R$ the region defined by $G$ and $\gate{G}{x}$ above $\gate{G}{x}$ and denote by $R'$ the region defined by $G'$ and $\gate{G'}{x}$ below $\gate{G'}{x}$.
    After splitting at $\gate{G}{z}$ and $\gate{G'}{z}$ (and ignoring $\gate{G}{x}$ and $\gate{G'}{x}$ again), we also end up with up to 5 regions.
    Two of them are convex by Observation \ref{obs:simple_case}.
    Note that $\gate{G}{z}$ has to intersect $R$, because otherwise, $\gate{G}{z}$ wouldn't minimize the distance in the $z$-portal graph to $G'$ as $T$ is simple.
    Similarly, $\gate{G'}{z}$ intersects $R'$.
    Thus, after splitting up $\gate{G}{z}$ and $\gate{G'}{z}$ at one of their nodes in the first step, $R$ and $R'$ are each split up into up to three parts.
    As before, one of these parts is a triangle and therefore trivially convex.
    The other two parts are defined by two gates that intersect in an obtuse angle and are thus convex by Observation \ref{obs:simple_case}.
    Hence, all regions are convex.
    The only remaining region is $M$, which is excluded from the lemma statement.
\hfill $\square$ 
\end{proof}

It remains to deal with region $M$, which might not yet be convex and needs further splits.
To this end, let $g$ be the closest node to $G$ on either $\gate{G}{x}$ or $\gate{G}{z}$ that is in $M$, i.e. $g := \argmin_{v \in (\gate{G}{x} \cap M) \cup (\gate{G}{z} \cap M)} d(v,G)$.
Note that if $M$ is enclosed by both $\gate{G}{x}$ and $\gate{G}{z}$, this is exactly the intersection point of $\gate{G}{x}$ and $\gate{G}{z}$, and if $M$ is enclosed by only one of $\gate{G}{x}$ or $\gate{G}{z}$, this is either $\boundarynode{G}{x}$ or $\boundarynode{G}{z}$.
Analogously, we define $g'\ := \argmin_{v \in (\gate{G'}{x} \cap M) \cup (\gate{G'}{z} \cap M)} d(v,G')$ on $G'$'s side.
By choice of $g$ and $g'$, we get the following crucial property: for any path between two nodes $u,v \in M$ that leaves $T$ through $G$ and reenters $T$ through $G'$, there is a $uv$-path of the same length that goes via $g$ and $g'$.
Hence, we can w.l.o.g. interpret $M$ as a tunnel region with single node gates $g$ and $g'$ and show how to split up this tunnel into convex regions, as the resulting regions remain convex even without this restriction.

Furthermore, we have the following property which will be important later.

\begin{lemma}\label{lem:g_not_on_G}
    It is $g \not\in G$ and $g' \not\in G'$.
\end{lemma}
\begin{proof}
    We only prove that $g \not\in G$, as $g' \not\in G'$ can be proved analogously.
    W.l.o.g., let $G$ be to the left of the tunnel region.
    The other case follows from symmetry.
    Now, assume for contradiction that $g \in G$.
    By definition of $g$, we have $g \in \{\boundarynode{G}{x}, \boundarynode{G}{z}\} \cup (\gate{G}{x} \cap \gate{G}{z})$.
    Since $\boundarynode{G}{x}, \boundarynode{G}{z} \not\in G$ by definition, it follows that $g \in \gate{G}{x} \cap \gate{G}{z}$, i.e. $\gate{G}{x}$ and $\gate{G}{z}$ have to intersect on $G$.
    Now, let $R_x(G')$ be the region in which $G'$ lies after (only) splitting at $\gate{G}{x}$ and let $R_z(G')$ be the region in which $G'$ lies after (only) splitting at $\gate{G}{z}$.
    Then, $\gate{G}{x}$ either touches a boundary node $\boundarynode{G}{x} \not\in G$ from the side of $R_x(G')$ or an endpoint of $G$ lies on $\gate{G}{x}$, as otherwise $\gate{G}{x}$ would not be the $x$-portal intersecting $G$ that minimizes the $x$-distance to $G'$.
    In the former case, we split at $\boundarynode{G}{x}$ and immediately get that $g \not\in M$, which is a contradiction.
    In the latter case, we consider two subcases.
    If the $\mathit{NNE}$-endpoint of $G$ lies on $\gate{G}{x}$, then $R_x(G')$ has to lie above $\gate{G}{x}$. 
    As the $\mathit{NNE}$-endpoint of $G$ lies on $\gate{G}{x}$, it follows that $\gate{G}{x}$ has to touch a boundary node $\boundarynode{G}{x} \not\in G$ from the side of $R_x(G')$. 
    Since we split at $\boundarynode{G}{x}$, this again implies that $g \not\in M$, which is a contradiction.
    Otherwise, if the $\mathit{SSW}$-endpoint of $G$ lies on $\gate{G}{x}$, we have that the $\mathit{SSW}$-endpoint of $G$ also lies on $\gate{G}{z}$.
    Then, $R_z(G')$ has to lie below $\gate{G}{z}$.
    As before, this implies that $\gate{G}{z}$ touches a boundary node $\boundarynode{G}{z} \not\in G$ from the side of $R_z(G')$.
    Since we split at $\boundarynode{G}{z}$, we get that $g \not\in M$, which is a contradiction.
    We reach a contradiction in any case and conclude that $g \not\in G$ has to hold.
\hfill $\square$ \end{proof}

For the remainder of this section, we show how to decompose the tunnel $M$ with single node gates $g$ and $g'$ into convex regions.
We do this by splitting $M$ at the portals $P_x$, $P_y$ and $P_z$, which we define in the following way (see Figure \ref{fig:convex_decomposition:point_shaped} for an example).

\begin{definition}
\label{def:point_shape_decomposition}
    Let $M$ be a tunnel region with two gates $g, g' \in M$ each consisting of a single node.
    For all directions $q \in \{x,y,z\}$, we split $M$ as follows.
    Let $d_q := d_{M,q}(g,g')$ be the $q$-distance between $g$ and $g'$ within $M$. 
    We split $M$ at the portal $P_q := \{v \in M \mid d_q(g,v) = \lceil \frac{d_q}{2} \rceil \wedge d_q(g',v) = \lfloor \frac{d_q}{2} \rfloor\}$.
    After splitting $M$ at (only) $P_q$, let $R_q(g)$ be the region in which $g$ lies and $R_q(g')$ be the region in which $g'$ lies.
    In case $R_q(g) = R_q(g')$, $P_q$ has to touch a boundary node $b_q$ from the side of $R_q(g)$ that lies on every shortest $gg'$-path within $M$.
    Then, we split $R_q(g)$ at the westernmost such node on $P_q$.
\end{definition}

\begin{lemma}
\label{lem:point_shaped_gates_are_portals}
    The sets $P_x, P_y$ and $P_z$ form a portal in $x$-, $y$- and $z$-direction in the region $M$ respectively. 
    Furthermore, each $gg'$-path within $M$ intersects with $P_x, P_y$ and $P_z$.
\end{lemma}
\begin{proof}
    We show the statement for $P_x$, as the other cases follow by symmetry.
    Consider the $x$-portal graph $\mathcal{P}_x$ of $M$.
    Since $M$ is simple, $\mathcal{P}_x$ is a tree by \Cref{lem:portal_tree}, and therefore, there is a unique path from $\operatorname{portal}_x(g)$ to $\operatorname{portal}_x(g')$ in $\mathcal{P}_x$. 
    Thus, all nodes $v \in M$ with $d_x(g,v) = \lceil \frac{d_x}{2} \rceil$ and $d_x(g',v) = \lfloor \frac{d_x}{2} \rfloor$ belong to the same portal $P_x$.
    Moreover, any $gg'$-path within $M$ intersects with $P_x$, as there is a unique path from $\operatorname{portal}_x(g)$ to $\operatorname{portal}_x(g')$ in $\mathcal{P}_x$ which intersects $P_x$. 
\hfill $\square$ 
\end{proof}

We will now prove that the regions created by splitting $M$ at $P_x$, $P_y$ and $P_z$ are indeed convex.
We start off by showing that for any two nodes $u,v$ of the same region $R$, a shortest path \textit{within} $M$ must also stay inside $R$.

\begin{lemma}
\label{lem:path_convex_inside_tunnel}
    Let $u,v \in R$ for some region $R \subseteq M$ resulting from splitting $M$ at $P_x, P_y$ and $P_z$. Then, $d_R(u,v) = d_M(u,v)$.
\end{lemma}
\begin{proof}
    Assume there is a shortest $uv$-path $\widetilde{\Pi}$ that is partially outside of $R$ but stays within $M$. 
    Path $\widetilde{\Pi}$ does not exit and re-enter $R$ through the same portal, as such a path can be made shorter by traversing that portal. 
    W.l.o.g. assume that $\widetilde{\Pi}$ exits $R$ through portal $P_x$ at node $v_x$ and re-enters $R$ through portal $P_y$ at node $v_y$. 
    As $\widetilde{\Pi}$ stays within $M$ and $M$ is simple, we can shorten $\widetilde{\Pi}$ by replacing the subpath from $v_x$ to $v_y$ with the simple path from $v_x$ to $v_y$ along the portals that make up the boundary of $R$.
    Therefore, $\widetilde{\Pi}$ is not a shortest $uv$-path, as it can be shortened by traversing that boundary instead of going around $R$. 
    As $\widetilde{\Pi}$ cannot exist, any shortest $uv$-path within $M$ also stays within $R$ and it is $d_R(u,v) = d_M(u,v)$.
\hfill $\square$ 
\end{proof}

\begin{corollary}
\label{cor:path_convex_inside_tunnel}
    Let $u,v \in R$ for some region $R \subseteq M$ resulting from splitting $M$ at $P_x, P_y$ and $P_z$. Then, $d_{R,x}(u,v) = d_{M,x}(u,v)$, $d_{R,y}(u,v) = d_{M,y}(u,v)$ and $d_{R,z}(u,v) = d_{M,z}(u,v)$.
\end{corollary}
\begin{proof}
    By Lemma \ref{lem:path_convex_inside_tunnel}, we know that there is a shortest $uv$-path $\Pi$ within $M$ that stays inside $R$. 
    As the portals of a fixed direction traversed by any shortest path correspond to a shortest path in the respective portal graph, the existence of $\Pi$ implies that there is a shortest $uv$-path within $M$ using only portals in $R$.
    
    \hfill $\square$ 
\end{proof}

Next, we show that nodes $u,v$ on two (possibly different) shortest $gg'$-paths within $M$, that are \textit{on the same side} of a portal $P_x$, $P_y$ or $P_z$, are not far away from each other.
In particular, Corollary \ref{cor:distance_bound_shortest_path_inside_region} implies that if $u,v$ are inside the same region, then the shortest $uv$-path stays within that region, because any $uv$-path leaving and reentering tunnel $M$ would be longer with respect to every direction.
These properties will allow us to show convexity for other pairs of nodes later on.

\begin{lemma}
\label{lem:distance_bound_shortest_path_same_side}
    Let $\Pi, \Pi'$ be two shortest $gg'$-paths within $M$. 
    Further, let $u$ on $\Pi$ and $v$ on $\Pi'$ such that both $u$ and $v$ are in the same region after splitting $M$ at (only) $P_q$ for a $q \in \{x,y,z\}$. 
    Then, $d_{M,q}(u,v) \leq \lceil \frac{d_q}{2} \rceil$.
\end{lemma}
\begin{proof}
    By \Cref{lem:point_shaped_gates_are_portals}, both $\Pi$ and $\Pi'$ intersect with $P_q$.
    Further, since $M$ is simple, the $q$-portal graph of $M$ is a tree by Lemma \ref{lem:portal_tree}.
    Therefore, both $\Pi$ and $\Pi'$ visit the same $q$-portals.
    We denote this sequence of $q$-portals as $(p_1, \dots, p_{k-1},\allowbreak P_q, p_{k+1},\allowbreak \dots, p_\ell)$ where $p_1 = \operatorname{portal}_q(g)$ and $p_\ell = \operatorname{portal}_q(g')$.
    Note that after splitting at $P_q$, no node of a portal in the subsequence $(p_1, \dots, p_{k-1})$ lies in the same region as any node of a portal in $(p_{k+1}, \dots, p_\ell)$.
    This clearly holds if $p_{k-1}$ and $p_{k+1}$ lie on opposite sides of $P_q$, as $P_q$ appears only once in the sequence and $M$ is simple.
    Otherwise, if $p_{k-1}$ and $p_{k+1}$ lie on the same side of $P_q$, $P_q$ has to touch a boundary node of $R_q(g)$ that lies on $\Pi$ and $\Pi'$.
    Since we split at that node, we get that $p_{k-1}$ and $p_{k+1}$ lie in different regions and the statement holds.
    Hence, $\operatorname{portal}_q(u)$ and $\operatorname{portal}_q(v)$ either both lie between $P_q$ and $\operatorname{portal}_q(g)$ or between $P_q$ and $\operatorname{portal}_q(g')$.
    From the definition of $P_q$, it follows that $d_{M,q}(u,v) \leq \lceil \frac{d_q}{2} \rceil$.
\hfill $\square$ 
\end{proof}

\begin{corollary}
\label{cor:distance_bound_shortest_path_inside_region}
    Let $\Pi, \Pi'$ be two shortest $gg'$-paths within $M$.
    Further, let $u$ on $\Pi$ and $v$ on $\Pi'$ be nodes in the same region $R \subseteq M$ resulting from splitting $M$ at $P_x, P_y$ and $P_z$. 
    Then, $d_{R,x}(u,v) \leq \lceil \frac{d_x}{2} \rceil$, $d_{R,y}(u,v) \leq \lceil \frac{d_y}{2} \rceil$ and $d_{R,z}(u,v) \leq \lceil \frac{d_z}{2} \rceil$.
\end{corollary}
\begin{proof}
    Since $u, v \in R$, both $u$ and $v$ are on the same sides of $P_x$ within $M$. 
    By \Cref{lem:distance_bound_shortest_path_same_side}, it is $d_{M,x}(u,v) \leq \lceil \frac{d_x}{2} \rceil$.
    As $d_{R,x}(u,v) = d_{M,x}(u,v)$ by Corollary \ref{cor:path_convex_inside_tunnel}, the statement follows.
    The same argument applies for the other two distance cases.
\hfill $\square$ \end{proof}

In order to argue about nodes $v$ that do not lie on shortest $gg'$-paths within $M$, we will consider shortest $gg'$-paths that are \emph{closest} to $v$.
To characterize the structure of these paths, we prove the following two general lemmas that we will use later.

\begin{lemma}
\label{lem:shortest_paths_stay_inside}
    Let $\Gamma = (V_\Gamma, E_\Gamma)$ be a simple triangular grid graph, let $u,v \in V_\Gamma$, and let $\mathcal{U}_{uv}$ be the set of nodes that are on a shortest $uv$-path in $\Gamma$. Then, $\mathcal{U}_{uv}$ is convex.
\end{lemma}
\begin{proof}
    We will prove the lemma by contradiction. 
    Let $x,y \in \mathcal{U}_{uv}$ and assume that a shortest $xy$-path $\Pi$ leaves $\mathcal{U}_{uv}$ directly after some node $x'$ and joins $\mathcal{U}_{uv}$ again directly before some node $y'$, while no nodes between $x'$ and $y'$ on $\Pi$ are in $\mathcal{U}_{uv}$. 
    Note that the boundary of $\mathcal{U}_{uv}$ consists of two shortest $uv$-paths $\Pi'$ and $\Pi''$. 
    If $x'$ and $y'$ are on the same path $\Pi'$, then $\Pi_{xx'} \circ \Pi'_{x'y'} \circ \Pi_{y',y}$ is shorter than $\Pi$, which contradicts $\Pi$ being a shortest path.
    If, however, $x'$ is on the path $\Pi'$ and $y'$ is on the other path $\Pi''$, then $\Pi$ eventually has to get around $u$ or $v$.
    Then, either $\Pi_{xx'} \circ \Pi'_{x'u} \circ \Pi''_{uy'} \circ \Pi_{y',y}$ or $\Pi_{xx'} \circ \Pi'_{x'v} \circ \Pi''_{vy'} \circ \Pi_{y',y}$ is shorter than $\Pi$, which again contradicts $\Pi$ being a shortest path.
\hfill $\square$ \end{proof}

\begin{lemma}
\label{lem:closest_path}
    Given a simple triangular grid graph $\Gamma = (V_\Gamma, E_\Gamma)$, let $a,b,v \in V_\Gamma$, and let $\Pi$ be a shortest $ab$-path closest to $v$ with the biggest number of closest points to $v$, i.e., for
    $$
        \hat{\Pi} \coloneqq \argmin_{\Pi' \text{ shortest $ab$-path}}\; \min_{u \in \Pi'} d_\Gamma(u,v)\text{, we obtain}\quad \Pi \in \argmax_{\Pi' \in \hat{\Pi}} |V'_{\Pi'}|
    $$
    where $V'_{\Pi'}$ denotes the set of closest points to $v$ on a path $\Pi'$, i.e., 
    $$
        V'_{\Pi'} \coloneqq \argmin_{u\in\Pi'} d_\Gamma(u,v)\text{.}
    $$
    Further, let $v'_a \coloneqq \argmin_{v'\in V'_\Pi} d_\Gamma(v',a)$ be the closest one of these nodes to $a$ and let\hspace{1pt} $v'_b :=\argmin_{v'\in V'_\Pi} d_\Gamma(v',b)$ be the closest one of these nodes to $b$.
    Both of the following hold:
    \begin{enumerate}
        \item All nodes of $V'_\Pi$ are continuous nodes on the same axis with endpoints $v'_a$ and $v'_b$.
        \item Every shortest $vv'_aa$-path is a shortest $va$-path and every shortest $vv'_bb$-path is a shortest $vb$-path.
    \end{enumerate}  
\end{lemma}
\begin{proof}
    We start with the proof of Property (1), i.e., we prove that the nodes of $V'_\Pi$ are continuous nodes on the same axis with endpoints $v'_a$ and $v'_b$.
    We denote the nodes of $V'_\Pi$ on $\Pi$ by $p_1, \dots, p_\ell$, where $\ell = |V'_\Pi|$.
    
    First, we prove that the nodes of $V'_\Pi$ are on the same axis by contradiction.
    Assume that the nodes of $V'_\Pi$ are not on the same axis.
    Let $B(v) := \{u \in V_\Gamma \mid d(u,v) \leq d(p_1,v)\}$ be all nodes of distance at most $d(p_1,v)$ from $v$ in $\Gamma$.
    As all nodes of $V'_\Pi$ have the same distance to $v$ and $p_1\in V'_\Pi$, it is $V'_\Pi \subseteq B(v)$ and all nodes $p \in V'_\Pi$ lie on the boundary of $B(v)$.
    Let $p_i,p_j \in V'_\Pi$ be two nodes on different axes.
    By definition of $B(v)$ and since $\Gamma$ is simple, for any two nodes $u_1, u_2 \in B(v)$ that are on different axes, there is a shortest path between $u_1$ and $u_2$ that goes through the interior of $B(v)$.
    In particular, this holds for $p_i$ and $p_j$.
    This, however, contradicts the fact that $\Pi$ is a shortest $ab$-path closest to $v$, as the nodes in the interior of $B(v)$ are closer to $v$ than the nodes on the boundary of $B(v)$. 
    Thus, the nodes of $V'_\Pi$ are on the same axis.

    Next, we prove that the nodes of $V'_\Pi$ are continuous by contradiction.
    Assume that the nodes of $V'_\Pi$ are not continuous, i.e., there are two nodes $p_i,p_j\in V'_\Pi$ such that the subpath $\Pi_{p_ip_j}$ from $p_i$ to $p_j$ of $\Pi$ has nodes not in $V'_\Pi$.
    Since $\Pi$ is a shortest $ab$-path closest to $v$, $\Pi_{p_ip_j}$ can not be in the interior of $B(v)$.
    Therefore, $\Pi_{p_ip_j}$ leave and reenter $B(v)$.
    Since the subgraph of $\Gamma$ induced by $B(v)$ is a simple subgraph of a hexagon, $B(v)$ is convex.
    Therefore, all shortest paths between $p_i$ and $p_j$ stay inside $B(v)$.
    This contradicts with $\Pi$ being a shortest $ab$-path and leaving and reentering $B(v)$.
    Thus, the nodes of $V'_\Pi$ are continuous.

    Finally, note that $v'_a$ (and $v'_b$) are clearly endpoints of the subpath of the nodes of $V'_\Pi$ on $\Pi$, as they are closest to $a$ (and $b$) among the nodes of $V'_\Pi$ by definition.
    
    \medskip

    We now prove Property (2), i.e., we show that every shortest $vv'_aa$-path is a shortest $va$-path and every shortest $vv'_bb$-path is a shortest $vb$-path.
    We will prove the statement for any shortest $vv'_aa$-path. 
    The shortest $vv'_bb$-path-case is analogous.
    In the following, we refer to the nodes of $\Pi$ as $v = p_1, \dots, p_\ell = v'_a$.
    Furthermore, let $\mathcal{U}_{av}$ denote the set containing all nodes on a shortest $av$-path, i.e. $\mathcal{U}_{av} \coloneqq \bigcup_{\text{$\widetilde{\Pi}$ shortest $av$-path}}\widetilde{\Pi}$.
    
    We start by proving by contradiction that $\Pi$ leaves $\mathcal{U}_{av}$ at most once, i.e. there is a unique point $p_i \in \Pi$ such that $p_{i+1} \not\in \mathcal{U}_{av}$.
    Assume $\Pi$ would leave $\mathcal{U}_{av}$ more than once and let $p_i$ be the first leaving node.
    Then there would be reentering node $p_j$ with $j>i$ s.t. $p_{j-1}\not\in\mathcal{U}_{av}$ but $p_j \in \mathcal{U}_{av}$.
    According to Lemma \ref{lem:shortest_paths_stay_inside}, all shortest $p_ip_j$-paths are inside $\mathcal{U}_{av}$.
    Therefore, $\Pi$ could be shortened by replacing its segment $\Pi_{p_ip_j}$ with such a shortest path. 
    This is a contradiction to $\Pi$ being a shortest path.

    We now call the unique point where $\Pi$ leaves $\mathcal{U}_{av}$ $\overline{v'}$.
    Our next goal is to show $\overline{v'} = v'_a$, i.e., to show that $\Pi$ leaves $\mathcal{U}_{av}$ precisely at the closest point to $v$ on $\Pi$ which is also closest to $a$.
    To this end, we show that $\overline{v'}$ can not be closer to $a$ than $v'_a$ and can not be further away from $a$ than $v'_a$.
    
    We start by proving the former by contradiction.
    Assume $d(\overline{v'},a) < d(v'_a,a)$. Since $\Pi$ is a shortest $ab$-path, the next node in $\Pi$ after $\overline{v'}$ must be closer to $v'_a$ than $\overline{v'}$. 
    As $v'_a$ is the first node closest to $v$ on $\Pi$ starting from $\overline{v'}$ and $\Pi$ contains the most closest points to $v$ among the paths closest to $v$, it must also be closer to $v$. 
    Hence it is part of a shortest $\overline{v'}v$-path and, since $\overline{v'} \in \mathcal{U}_{av}$, it is part of $\mathcal{U}_{av}$.
    This is a contradiction to the definition of $\overline{v'}$, as $\overline{v'}$ would not be the point where $\Pi$ leaves $\mathcal{U}_{av}$.

    Assume $d(\overline{v'},a) > d(v'_a,a)$. By definition of $v'_a$, the node after $v'_a$ on $\Pi$ has at least the same distance to $v$ and a larger distance to $a$.
    Therefore, it can not be part of $\mathcal{U}_{av}$, as a shortest $vv'_aa$-path is shorter than a $va$-path via that node.
    This, however, contradicts the definition of $\overline{v'}$, as $\overline{v'}$ would not be the point where $\Pi$ leaves $\mathcal{U}_{av}$.

    Hence $d(\overline{v'},a) = d(v'_a,a)$. Since both $\overline{v'}$ and $v'_a$ are on the shortest $ab$-path $\Pi$, this yields $\overline{v'} = v'_a$. 
    By definition of $\overline{v'}$, we conclude $v'_a \in \mathcal{U}_{av}$, i.e. $v'_a$ is contained in the set of shortest paths from $a$ to $v$. 
    Thus, a shortest $vv'_aa$-path is a shortest $va$-path.
\hfill $\square$ \end{proof}

We are now ready to prove that the regions created by splitting $M$ at $P_x$, $P_y$ and $P_z$ are convex.
Coy \emph{et al.} \cite{DBLP:journals/tcs/CoyCSSW24} exploited that in a square grid graph, for any node $u$ inside a region $R$, there is a node $u'_g \in R$, such that $u'_g$ is on a shortest $gg'$-path closest to $u$ within $M$ and on a shortest $ug$-path.
For triangular grid graphs, this property does not hold since there can be regions that are not crossed by any shortest $gg'$-path within $M$.
Thus, such a $u'_g$ clearly cannot exist for nodes $u$ from such a region.
Moreover, even if a region $R$ is crossed by a shortest $gg'$-path $\Pi$ within $M$, such a $u'_g$ doesn't necessarily exist, since by \Cref{lem:closest_path}, a shortest $ug$-path only connects with $\Pi$ at the node \textit{closest} to $g$ among the nodes of $\Pi$ closest to $u$.
This closest node, however, might already be outside of $R$.
Therefore, we have to explore in more detail what happens in the area \emph{between} region $R$ and where a shortest $gg'$-path meets the shortest $ug$- (or $vg'$-)path.
Ultimately, we reach the following result, which concludes the decomposition.

\begin{lemma}
\label{lem:point_shape_tunnel_path_convex}
    Let $R \subseteq M$ be a region resulting from splitting $M$ at $P_x, P_y$ and $P_z$.
    Then, $R$ is convex.
\end{lemma}
\begin{proof}
    \begin{figure}[!ht]
      \centering
      \subfloat{\label{fig:lemma56:a}
      \includegraphics[width=0.7\textwidth]{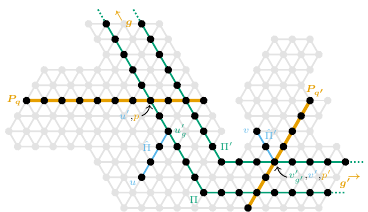}}\quad
      \subfloat{\label{fig:lemma56:b}
      \includegraphics[width=0.7\textwidth]{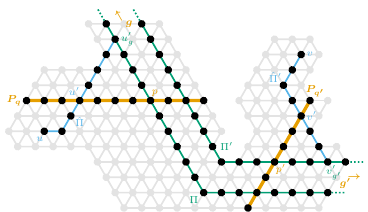}}\quad
      \caption{Examples of the notation used in the proof of Lemma \ref{lem:point_shape_tunnel_path_convex}. Region $R$ intersects with $P_x = P_q$ and $P_y = P_{q'}$. In the top figure, both $u'_g, v'_{g'} \in R$. In the bottom figure, both $u'_g, v'_{g'} \not\in R$.}
      \label{fig:lemma56}
    \end{figure}
    Let $u,v \in R$ and let $\widetilde{\Pi}$ be a $uv$-path that is partially outside of $R$.
    We will prove that $|\widetilde{\Pi}| > d_R(u,v)$ by showing that there is another path that stays within $R$ and is shorter than $\widetilde{\Pi}$. 
    We will assume w.l.o.g. that neither $g$ nor $g'$ lie in $R$, as this is the more interesting case.
    The other case can be argued analogously.
    
    First, assume that $\widetilde{\Pi}$ leaves $R$ but not $M$.
    Then, as argued in Lemma \ref{lem:path_convex_inside_tunnel}, $\widetilde{\Pi}$ can be shortened by staying within $M$.
    Thus, $|\widetilde{\Pi}| > d_R(u,v)$ in this case.
    Therefore, we only have to consider paths $\widetilde{\Pi}$ that leave $M$.
    Furthermore, we only have to consider paths $\widetilde{\Pi}$ that connect $u$ to one of the two gates of $M$ and $v$ to the other gate of $M$ via shortest paths, as no shortest path leaves and reenters $M$ through the same gate.
    W.l.o.g. assume that $\widetilde{\Pi}$ first connects $u$ with $g$ on a shortest path and later reenters $R$ to connect $g'$ with $v$ on a shortest path.
    Let $\Pi$ and $\Pi'$ each be a shortest $gg'$-path within $M$ that is closest to $u$ and $v$ respectively and that also has the biggest number of closest points to $u$ and $v$ respectively.
    Now, let $u'_g$ be the closest node to $g$ on $\Pi$ among the nodes of $\Pi$ closest to $u$ (see \Cref{fig:lemma56}).
    Note that $u'_g \not\in R$ may hold.
    Further, let $\hat{\Pi}$ be a shortest path from $u$ to $u'_g$ that stays inside $R$ as long as possible.
    By \Cref{lem:closest_path}, we know that $\hat{\Pi} \circ \Pi_{u'_gg}$ is a shortest $ug$-path.
    Let $u'$ be the last node on $\hat{\Pi} \circ \Pi_{u'_gg}$ within $R$.
    Since $g \not\in R$, $u'$ has to lie on (at least) one of the portals $P_x, P_y$ or $P_z$.
    Let $P_q$ with $q \in \{x,y,z\}$ be a portal such that $u' \in P_q$ (if there is more than one such portal, pick one arbitrarily).
    We define analogous nodes on $v$'s side: let $v'_{g'}$ be the closest node to $g'$ on $\Pi'$ among the nodes of $\Pi'$ closest to $v$ and let $\hat{\Pi}'$ be a shortest path from $v$ to $v'_{g'}$ that stays inside $R$ as long as possible.
    Let $v'$ be the last node within $R$ on the shortest $vg$-path $\hat{\Pi}' \circ \Pi'_{v'_{g'}g'}$.
    Note that $v'$ has to lie on (at least) one of the portals $P_x, P_y$ or $P_z$ that is different from $P_q$, as otherwise one of the paths $\hat{\Pi} \circ \Pi_{u'_gg}$ or $\hat{\Pi}' \circ \Pi'_{v'_{g'}g'}$ would cross $P_q$ twice which contradicts them being shortest paths.
    Hence, let $P_{q'}$ with $q' \in \{x,y,z\}$ be a portal such that $v' \in P_{q'}$ and $q \not= q'$ (if there is more than one such portal, pick one arbitrarily).
    We denote the remaining direction by $q'' \in \{x,y,z\} \setminus\{q,q'\}$.
    See \Cref{fig:lemma56} for an example.

    With these definitions in place, we will first argue that $u'$ and $u'_g$ have to lie on the same side of $P_{q''}$.
    This is clearly true if $u'_g \in R$, since we have $u' \in R$ by definition and all nodes of the same region lie on the same side of all portals.
    It remains to consider the case that $u'_g \not\in R$.
    Let $p$ be the node in $\Pi \cap P_q$ that is closest to $u'$.
    Since $u'_g$ is a closest node on $\Pi$ to $u$ and $u'$ is on a shortest $uu'_g$-path, $u'_g$ is also a closest node on $\Pi$ to $u'$.
    Further, since $p$ is the closest node to $u'$ in $\Pi \cap P_q$, since $\Pi$ maximizes the number of closest points to $u$ and since $M$ is simple, $p$ is also a closest node on $\Pi$ to $u'$.
    Recall that by \Cref{lem:closest_path}, the subpath $\Pi_{u'_gp}$ is a straight line.
    This implies that $\hat{\Pi}_{u'u'_g}$, $\Pi_{u'_gp}$ and the path between $p$ and $u'$ on $P_q$ together form an equilateral triangle.
    Now, assume for contradiction that $u'$ and $u'_g$ lie on different sides of $P_{q''}$, i.e. $P_{q''}$ cuts through that triangle between $u'$ and $u'_g$.
    Then, $\hat{\Pi}_{u'u'_g}$ is a straight line on the $q'$-axis and $\Pi_{u'_gp}$ is a straight line on the $q''$-axis.
    Since $\Pi$ is a shortest $gg'$-path, one of the subpaths $\Pi_{gu'_g}$ and $\Pi_{pg'}$ has to intersect $P_{q''}$.
    Let $a \in \Pi \cap P_{q''}$ and let $b \in \hat{\Pi}_{u'u'_g} \cap P_{q''}$.
    If $a \in \Pi_{gu'_g}$, then we can identify a shortest $ap$-path by traversing $P_{q''}$ and $P_q$.
    Hence, we can replace the subpath $\Pi_{ap}$ with this new path to obtain another shortest $gg'$-path.
    However, since $b$ lies on this path and $b$ is closer to $u'$ than $u'_g$ (and thus also closer to $u$ than $u'_g$), this contradicts $\Pi$ being a shortest $gg'$-path closest to $u$.
    Similarly, if $a \in \Pi_{pg'}$, we can identify a shortest $u'_ga$-path by traversing $\hat{\Pi}_{u'_gb}$ and $P_{q''}$.
    This again yields a shortest $gg'$-path that is closer to $u$ than $\Pi$.
    We reach a contradiction in both cases and conclude that indeed both $u'$ and $u'_g$ lie on the same side of $P_{q''}$.
    Analogously, we can argue that both $v'$ and $v'_{g'}$ lie on the same side of $P_{q''}$.
    Since both $u', v' \in R$ by definition, we get that $u'_g$ and $v'_{g'}$ lie on the same side of $P_{q''}$.
    We call this property ($\Delta$).

    We will now upper bound the $q$-, $q'$- and $q''$-distances between $u'$ and $v'$ in $R$.
    We begin with the $q$-direction.
    First, note that it is $d_{M,q}(u',g') = d_{M,q}(p,g')$ because both $u',p \in P_q$.
    Moreover, it is $d_{M,q}(p, g') \leq \lceil \frac{d_q}{2} \rceil$ by \Cref{lem:distance_bound_shortest_path_same_side} since both $p$ and $g'$ lie on a shortest $gg'$-path and are on the same side of $P_q$.
    Combining both facts we get $d_{M,q}(u',g') \leq \lceil \frac{d_q}{2} \rceil$.
    Altogether, we conclude:
    \[d_{R,q}(u',v') \overset{\text{Cor. } \ref{cor:path_convex_inside_tunnel}}{=} d_{M,q}(u',v') \leq d_{M,q}(u',g') + d_{M,q}(g', v') \leq \lceil \frac{d_q}{2} \rceil + d_{M,q}(v',g').\]

    We bound the $q'$-distance between $u'$ and $v'$ in $R$ similarly.
    Since both $v',p' \in P_{q'}$, we get $d_{M,q'}(g,v') = d_{M,q'}(g,p')$.
    Moreover, it is $d_{M,q'}(g,p') \leq \lceil \frac{d_{q'}}{2} \rceil$ by \Cref{lem:distance_bound_shortest_path_same_side} since both $p'$ and $g$ lie on a shortest $gg'$-path and are on the same side of $P_{q'}$.
    Combining both facts we get $d_{M,q'}(g,v') \leq \lceil \frac{d_{q'}}{2} \rceil$.
    Altogether, we conclude:
    \[d_{R,q'}(u',v') \overset{\text{Cor. }\ref{cor:path_convex_inside_tunnel}}{=} d_{M,q'}(u',v') \leq d_{M,q'}(u',g) + d_{M,q'}(g, v') \leq d_{M,q'}(u',g) + \lceil \frac{d_{q'}}{2} \rceil.\]
    Finally, for the $q''$-distance between $u'$ and $v'$ in $R$, we know by property $(\Delta)$ that both $u'_g$ and $v'_{g'}$ lie on the same side of $P_{q''}$.
    Using \Cref{lem:distance_bound_shortest_path_same_side}, we get:
    \begin{align*}
        d_{R,q''}(u',v') \overset{\text{Cor. }\ref{cor:path_convex_inside_tunnel}}{=} d_{M,q''}(u',v') &\leq d_{M,q''}(u',u'_g) + d_{M,q''}(u'_g, v'_{g'}) + d_{M,q''}(v'_{g'},v')\\ &\overset{\text{Lem. }\ref{lem:distance_bound_shortest_path_same_side}}{\leq} d_{M,q''}(u',u'_g) + \lceil \frac{d_{q''}}{2} \rceil + d_{M,q''}(v'_{g'},v').
    \end{align*}
    We call the three inequalities above property (A).

    Next, we lower bound $d_M(u',g) + d_M(v',g')$.
    Since $u' \in P_q$ and $v' \in P_{q'}$, it is $d_{M,q}(u',g) \geq \lceil \frac{d_q}{2} \rceil$ and $d_{M,q'}(v',g') \geq \lfloor \frac{d_{q'}}{2} \rfloor$ by the definitions of $P_q$ and $P_{q'}$.
    We get the following lower bounds for the $q$- and $q'$-directions:
    \begin{align*}
        d_{M,q}(u',g) + d_{M,q}(v',g') &\geq \lceil \frac{d_q}{2} \rceil + d_{M,q}(v',g')\\
        d_{M,q'}(u',g) + d_{M,q'}(v',g') &\geq d_{M,q'}(u',g) + \lfloor \frac{d_{q'}}{2} \rfloor   
    \end{align*}
    Furthermore, we know by property $(\Delta)$ that one of the shortest paths $\Pi_{u'_gg}$ or $\Pi'_{v'_{g'}g'}$ crosses $P_{q''}$.
    Therefore, it is $d_{M,q''}(u'_g,g) + d_{M,q''}(v'_{g'},g') \geq \lfloor \frac{d_{q''}}{2} \rfloor$ by the definition of $P_{q''}$.
    We get the following lower bound for the $q''$-direction:
    \begin{align*}
        d_{M,q''}(u',g) + d_{M,q''}(v',g') &= d_{M,q''}(u',u'_g) + d_{M,q''}(u'_g, g) + d_{M,q''}(v',v'_{g'}) +\\&\hspace{6.1cm} d_{M,q''}(v'_{g'},g')\\ &\geq d_{M,q''}(u',u'_g) + \lfloor \frac{d_{q''}}{2} \rfloor + d_{M,q''}(v',v'_{g'})
    \end{align*}

    We call the three inequalities above property (B).

    Plugging it all together, we get for the part $\widetilde{\Pi}_M$ of $\widetilde{\Pi}$ inside $M$:
    \begin{align*}
   		|\widetilde{\Pi}_M|   &\geq d_M(u,g) + d_M(v,g')\\
   		&= d_R(u,u') + d_M(u',g) + d_R(v,v') + d_M(v',g')\\
   		&\overset{(B),\text{Lem. }\ref{lem:portal_distances}}{\geq} d_R(u,u') + d_R(v,v') + \frac{1}{2}\Big( \lceil \frac{d_q}{2} \rceil + d_{M,q}(v',g') +\\ &\hspace{4.5cm}d_{M,q'}(u',g) + \lfloor \frac{d_{q'}}{2} \rfloor +\\ &\hspace{4.5cm} d_{M,q''}(u',u'_g) + \lfloor \frac{d_{q''}}{2} \rfloor + d_{M,q''}(v'_{g'},v') \Big)\\
   		&= d_R(u,u') + d_R(v,v') + \frac{1}{2}\Big( \lceil \frac{d_q}{2} \rceil + d_{M,q}(v',g') +\\ &\hspace{4.5cm} d_{M,q'}(u',g) + \lceil \frac{d_{q'}}{2} \rceil -\mathbbm{1}_{\text{[$d_{q'}$ odd]}} +\\ &\hspace{4.5cm} d_{M,q''}(u',u'_g) + \lceil \frac{d_{q''}}{2} \rceil - \mathbbm{1}_{\text{[$d_{q''}$ odd]}} +\\&\hspace{8.2cm} d_{M,q''}(v'_{g'},v') \Big)\\
   		&= d_R(u,u') + d_R(v,v') + \frac{1}{2}\Big( \lceil \frac{d_q}{2} \rceil + d_{M,q}(v',g') +\\ &\hspace{4.5cm} d_{M,q'}(u',g) + \lceil \frac{d_{q'}}{2} \rceil +\\ &\hspace{4.5cm} d_{M,q''}(u',u'_g) + \lceil \frac{d_{q''}}{2} \rceil + d_{M,q''}(v'_{g'},v') \Big)\\ &\hspace{3.6cm} -\frac{\mathbbm{1}_{\text{[$d_{q'}$ odd]}} + \mathbbm{1}_{\text{[$d_{q''}$ odd]}}}{2}\\
   		&\overset{(A)}{\geq} d_R(u,u') + d_R(v,v') + \frac{1}{2}\Big(d_{R,q}(u',v') + d_{R,q'}(u',v') + d_{R,q''}(u',v')\Big)\\
   		&\hspace{3.6cm}- \frac{\mathbbm{1}_{\text{[$d_{q'}$ odd]}} + \mathbbm{1}_{\text{[$d_{q''}$ odd]}}}{2}\\
   		&\overset{Lem. \ref{lem:portal_distances}}{=} d_R(u,u') + d_R(v,v') + d_R(u',v') - \frac{\mathbbm{1}_{\text{[$d_{q'}$ odd]}} + \mathbbm{1}_{\text{[$d_{q''}$ odd]}}}{2}\\
   		&\geq d_R(u,v) - \frac{\mathbbm{1}_{\text{[$d_{q'}$ odd]}} + \mathbbm{1}_{\text{[$d_{q''}$ odd]}}}{2}
    \end{align*}
    Note that we lower bounded the part of $\widetilde{\Pi}$ inside $M$. 
    By \Cref{lem:g_not_on_G}, we know that the distance between $g$ and $g'$ outside of $M$ is at least $2$.
    Hence, we have $|\widetilde{\Pi}| \geq |\widetilde{\Pi}_M| + 2 > d_R(u,v)$.
    We conclude that $R$ is indeed convex.
\hfill $\square$ 
\end{proof}


\section{Implementation in the Amoebot Model}

In this section, we show how we compute the convex decomposition in the geometric amoebot model \cite{DBLP:journals/dc/DaymudeRS23,DBLP:conf/spaa/DerakhshandehDGRSS14} with reconfigurable circuits \cite{DBLP:journals/jcb/FeldmannPSD22}, which we describe in \Cref{sec:amoebot:model}.
In \Cref{sec:amoebot:algorithms,sec:amoebot:global_maxima}, we introduce algorithms from previous work that we will use as subroutines.
Finally, in \Cref{sec:amoebot:decomposition}, we present our convex decomposition algorithm.


\subsection{Geometric Amoebot Model with Reconfigurable Circuits}
\label{sec:amoebot:model}

We now formally introduce the \emph{geometric amoebot model}.
We will explain the model to a level of detail that is sufficient to understand the results of this paper.
For all other (unused) features of the model, e.g., movements, we refer to \cite{DBLP:journals/dc/DaymudeRS23,DBLP:conf/spaa/DerakhshandehDGRSS14}.
The model places a set of $n$ anonymous finite state machines called \emph{amoebots} with constant memory on the infinite regular triangular grid graph $G_\Delta = (V_\Delta, E_\Delta)$.
Each amoebot occupies one node and every node is occupied by at most one amoebot.
We assume that all amoebots have the same compass orientation (it defines one of its incident edges in $G_\Delta$ as the eastern direction) and chirality.%
\footnote{This is reasonable under the considered communication model since Feldmann \emph{et al.} \cite{DBLP:journals/jcb/FeldmannPSD22} showed that all amoebots are able to quickly come to an agreement.}
Let the \emph{amoebot structure} $V_\Gamma \subseteq V_\Delta$ be the set of nodes occupied by the amoebots.
By abuse of notation, we identify amoebots with their nodes.
We assume that $G_\Gamma = (V_\Gamma, E_\Gamma)$ is connected, where $G_\Gamma = G|_{V_\Gamma}$ is the graph induced by $V_\Gamma$.

We utilize the \emph{reconfigurable circuit extension} by Feldmann \emph{et al.} \cite{DBLP:journals/jcb/FeldmannPSD22} as our communication model.
In this extension, each edge between two neighboring amoebots $u$ and $v$ is replaced by $c$ edges called \emph{external links} with endpoints called \emph{pins}, for some constant $c \geq 1$ that is the same for all amoebots.
For each of these links, one pin is owned by $u$ while the other pin is owned by $v$.
In this paper, we assume that neighboring amoebots have a common labeling of their incident external links.

Each amoebot $u$ \emph{partitions} its \emph{pin set} $P(u)$ into a collection $\mathcal C(u)$ of pairwise disjoint subsets such that the union equals the pin set, i.e., $P(u) = \bigcup_{C \in \mathcal C(u)} C$.
We call $\mathcal C(u)$ the \emph{pin configuration} of $u$ and $C \in \mathcal C(u)$ a \emph{partition set} of $u$.
Let $\mathcal C = \bigcup_{u \in S} \mathcal C(u)$ be the collection of all partition sets in the system.
Two partition sets are \emph{connected} if and only if there is at least one external link between those sets.
Let $L$ be the set of all connections between the partition sets in the system.
Then, we call $H=(\mathcal C,L)$ the \emph{pin configuration} of the system and any connected component $C$ of $H$ a \emph{circuit}.
An amoebot is part of a circuit if and only if the circuit contains at least one of its partition sets.
A priori, an amoebot $u$ may not know whether two of its partition sets belong to the same circuit or not since initially it only knows $\mathcal C(u)$.

Each amoebot $u$ can send a primitive signal (a \emph{beep}) via any of its partition sets $C \in \mathcal C(u)$ that is received by all partition sets of the circuit containing $C$ at the beginning of the next round.
The amoebots are able to distinguish between beeps arriving at different partition sets.
More specifically, an amoebot receives a beep at partition set $C$ if at least one amoebot sends a beep on the circuit belonging to $C$, but the amoebots neither know the origin of the signal nor the number of origins.

We assume the fully synchronous activation model, i.e., the time is divided into synchronous rounds, and every amoebot is active in each round.
On activation, each amoebot may update its state, reconfigure its pin configuration, and activate an arbitrary number of its partition sets.
The beeps are propagated on the updated pin configurations.
The time complexity of an algorithm is measured by the number of synchronized rounds required by it.



\subsection{Subroutines from Prior Work}
\label{sec:amoebot:algorithms}


In this section, we list algorithms from previous work on the model that we will use as a blackbox in our decomposition algorithm.



\begin{theorem}[Adapted from \cite{DBLP:journals/nc/PadalkinSW24}]
\label{th:leader_election}
    Let $C_1, \dots, C_m$ be sets of candidates such that each set $C_i$ is connected by a unique circuit $\mathcal C_i$.
    The amoebot structure can elect a leader from each set of candidates within $\Theta(\log n)$ rounds w.h.p.
\end{theorem}

\begin{remark}
\label{rm:leader_election}
    We can use \Cref{th:leader_election} to elect leaders in a portal graph by selecting a representative for each portal, e.g., the northermost amoebot of each $y$-portal.
\end{remark}


\begin{theorem}[Adapted from \cite{DBLP:journals/nc/PadalkinSW24}]
\label{th:boundary_test}
    An amoebot structure can determine for each boundary set whether it is an inner boundary set or an outer boundary set within $O(\log n)$ rounds w.h.p.
\end{theorem}


Padalkin and Scheideler proposed various tree primitives utilizing reconfigurable circuits.
They have also shown how an amoebot structure can simulate these on portal trees.


\begin{theorem}[Adapted from \cite{DBLP:conf/podc/PadalkinS24}]
    \label{th:tp:pasc}
    Let $\mathcal P = (V_{\mathcal P}, E_{\mathcal P})$ be a portal tree.
    Let $R \in V_{\mathcal P}$.
    Then, the \emph{PASC algorithm} lets each portal $P \in V_{\mathcal P}$ compute distances $d(R, P)$ and $m = \max_{P' \in V_{\mathcal P}} d(R, P')$ within $O(\log m) = O(\log n)$ rounds.%
    \footnote{Note that in general a portal is not able to store its distance since amoebots only have constant sized memory. Hence, the computation happens in iterations. In each iteration, each portal computes a single bit of the two distances. We are still able to perform simple operations in an online fashion, e.g., comparisons and bit shifts (in order to divide by $2$) \cite{DBLP:conf/wdag/ArtmannPS25}.}
\end{theorem}


\begin{theorem}[Adapted from \cite{DBLP:conf/podc/PadalkinS24}]
\label{th:tp:rap}
    Let $\mathcal P = (V_{\mathcal P}, E_{\mathcal P})$ be a portal tree.
    Let $\mathcal Q \subseteq V_{\mathcal P}$ and $R \in \mathcal Q$.
    Then, the \emph{root and prune primitive} roots $\mathcal P$ at $R$ and prunes all subtrees without a portal in $\mathcal Q$ within $O(\log |\mathcal Q|) = O(\log n)$ rounds.
\end{theorem}


Finally, we list methods for portals and regions which were also introduced by Padalkin and Scheideler.

\begin{lemma}[Adapted from \cite{DBLP:conf/podc/PadalkinS24}]
\label{lem:tp:deg}
    Let $\mathcal P = (V_{\mathcal P}, E_{\mathcal P})$ be a portal tree.
    Let $\deg(P)$ denote the degree of $P$.
    Let $c$ be a constant.
    Then, each portal $P$ can check whether $\deg(P) \geq c$ within $O(\min\{ c, \log \deg_\mathcal P (P) \}) = O(1)$ rounds.
\end{lemma}

\begin{lemma}[Adapted from \cite{DBLP:conf/podc/PadalkinS24}]
\label{lem:tp:property:region}
    Let $R \subseteq V_\Gamma$ be a connected subset, e.g., a region or a portal.
    Let $S \subseteq V_\Gamma$.
    Then, region $R$ can determine whether $S \cap R \neq \emptyset$ within $O(1)$ rounds.
\end{lemma}

\begin{lemma}[Adapted from \cite{DBLP:conf/podc/PadalkinS24}]
\label{lem:tp:unmark}
    Let $P$ be a portal.
    Let $S \subseteq P$ be a non-empty set of amoebots and $u \in P$ one of the endpoints of the portal.
    Then, we can compute the closest amoebot $v \in S$ to $u$ in $O(1)$ rounds.
\end{lemma}



\subsection{Global Maxima}
\label{sec:amoebot:global_maxima}

Let $d$ be a cardinal direction and let $R \subseteq V_\Gamma$ be a non-empty set of amoebots.
We call $\argmin_{w \in R} \operatorname f_d(R,w)$ the \emph{global maxima} of $R$ in direction $d$ where $\operatorname f_d(R,w)$ denotes the number of amoebots in $R$ that lie in direction $d$ from amoebot $w$.

\begin{theorem}[Adapted from \cite{DBLP:journals/nc/PadalkinSW24}]
\label{th:global_maxima}
    The amoebot structure can compute the global maxima of $R$ within $O(\log^2 n)$ rounds w.h.p.
\end{theorem}

The idea of their algorithm is to first elect an arbitrary reference amoebot $u \in R$, then apply a spatial version of the PASC algorithm to compute the distance to that reference amoebot with respect to the cardinal direction $d$ and finally apply a consensus algorithm to determine the amoebots with the maximal distance.
However, the PASC algorithm computes the distances from the least significant bit to the most significant bit while the consensus algorithm requires the distances from the most significant bit to the least significant bit.
Since the amoebots cannot store their distances, they have to recompute the distances for each iteration of the consensus algorithm.

We can improve the runtime for boundary sets as follows.
First, we apply \Cref{th:leader_election} to split the boundary cycle into a chain.
Then, we apply the block primitive of \cite{DBLP:journals/nc/PadalkinSW24} to divide the chain into blocks of length $\Theta(\log n)$.
Next, each block applies \Cref{th:global_maxima} to compute its global maxima in $R$.
Let $M \subseteq R$ denote the set of all global maxima of all blocks.
Finally, we apply \Cref{th:global_maxima} to compute the global maxima of $M$.
However, each amoebot in $M$ now stores its distance in its block such that we do not need to recompute the distances for each iteration of the consensus algorithm.

\begin{theorem}
\label{th:global_maxima:boundary_set}
    A boundary set can compute the global maxima of $R$ within $O(\log n)$ rounds w.h.p.
\end{theorem}

\begin{proof}
    If an amoebot is a global maximum of $R$, then it is also a global maximum of its block.
    Hence, the final application of the global maxima algorithm also computes the global maxima of $R$.
    Note that each block is able to store the distance since it has $\Theta(\log n)$ amoebots.
    If a block has more than one global maxima, they must have the same distance such that each block has to only store one distance.

    By \Cref{th:leader_election}, the splitting of the boundary cycle into a chain requires $O(\log n)$ rounds w.h.p.
    The block primitive requires $O(\log n)$ rounds.
    We refer to \cite{DBLP:journals/nc/PadalkinSW24} for the details.
    Since the blocks have size $\Theta(\log n)$, the first application of \Cref{th:global_maxima} requires $O(\log^2 \log n) = O(\log n)$ rounds w.h.p.
    Since we do not have to recompute for each iteration of the consensus algorithm, the second application of \Cref{th:global_maxima} requires $O(\log n)$ rounds w.h.p.
    Overall, the algorithm requires $O(\log n)$ rounds w.h.p.
\hfill $\square$ 
\end{proof}

\begin{remark}
    \Cref{th:global_maxima:boundary_set} immediately improves the runtime of the spanning tree algorithm of Padalkin \emph{et al.} \cite{DBLP:journals/nc/PadalkinSW24} to $O(\log n)$ rounds w.h.p.
\end{remark}

If $R$ is connected, the global maxima must be on the outer boundary.
Hence, \Cref{th:boundary_test,th:global_maxima:boundary_set} imply the following corollary.

\begin{corollary}
\label{cor:global_maxima:connected}
    If $R$ is connected, the amoebot structure can compute the global maxima of $R$ within $O(\log n)$ rounds w.h.p.
\end{corollary}



\subsection{Convex Decomposition}
\label{sec:amoebot:decomposition}

In this section, we show how an amoebot structure can compute the convex decomposition described in \Cref{sec:pathconvex_region_decomposition}.


\subsubsection{Decomposition into $\Theta(|\mathcal{H}|)$ simple regions}
First, each boundary set determines whether it is an inner or outer boundary set (see \Cref{th:boundary_test}).
Second, each inner boundary set computes the $\mathit{NNE}$-most amoebot of the $\mathit{WNW}$-most amoebots and the $\mathit{NNE}$-most amoebot of the $\mathit{ESE}$-most amoebots (see \Cref{th:global_maxima:boundary_set}).
Let $S$ denote the set of amoebots computed by the inner boundary sets.
Third, each $y$-portal $P$ determines whether $S \cap P \neq \emptyset$ (see \Cref{lem:tp:property:region}).
Let $\mathcal P = \{P \in V_{\mathcal P_y} \mid S \cap P \neq \emptyset \}$.
We split the amoebot structure at each portal in $\mathcal P$ and each amoebot in $S$.

\begin{lemma}
\label{lem:amoebot:simple}
    An amoebot structure computes a decomposition consisting of $\Theta(|\mathcal{H}|)$ simple regions within $O(\log n)$ rounds w.h.p.
\end{lemma}

\begin{proof}
    The correctness follows from \Cref{lem:decomposition:simple}, and the runtime follows from \Cref{th:boundary_test,th:global_maxima:boundary_set,lem:tp:property:region}.
\hfill $\square$ 
\end{proof}


\subsubsection{Decomposition into $\Theta(|\mathcal{H}|)$ simple tunnel regions}
We apply the decomposition algorithm by Padalkin and Scheideler \cite{DBLP:conf/podc/PadalkinS24}.
For the sake of completeness, we outline the algorithm in the following.
First, each region applies a leader election on the gates it intersects (see \Cref{th:leader_election,rm:leader_election}).
This allows us to apply the root and prune primitive to prune all subtrees without a gate (see \Cref{th:tp:rap}).
Then, we split the region at all non-gate portals of degree at least $3$ (see \Cref{lem:tp:deg}).
Finally, we split each gate of degree at least $2$ as follows.
The gate marks the northernmost amoebot adjacent to each adjacent portal.
Note that each amoebot can locally decide whether it is one of these amoebots.
Then, the gate removes the northernmost marked amoebot and splits the region at the remaining marked amoebots (see \Cref{lem:tp:unmark}).

\begin{lemma}[Adapted from \cite{DBLP:conf/podc/PadalkinS24}]
\label{lem:amoebot:tunnel}
    Given a decomposition consisting of $\Theta(|\mathcal{H}|)$ simple regions,
    an amoebot structure computes a decomposition consisting of $\Theta(|\mathcal{H}|)$ simple tunnel regions within $O(\log n)$ rounds w.h.p.
\end{lemma}


\subsubsection{Decomposition into $\Theta(|\mathcal{H}|)$ simple, geodesically convex regions}
First, each tunnel region that intersects two gates applies a leader election on the gates it intersects (see \Cref{th:leader_election,rm:leader_election}).
Let $G$ be the elected gate and $G'$ the other one.
Let $\operatorname{portal}_q(S) = \{ \operatorname{portal}_q(u) \mid u \in S \}$ for a set $S$ of amoebots.
Now, for $q \in \{ x, z \}$, each $q$-portal determines whether it is in $\operatorname{portal}_q(G)$ or $\operatorname{portal}_q(G')$ (see \Cref{lem:tp:property:region}) and shares this information with adjacent portals.
This allows portals $P_{\uparrow,q}$ and $P_{\downarrow,q}$ to identify themselves if they exist.
If they exist, they identify $b_{\uparrow,q}$ ($b_{\downarrow,q}$) by determining their westernmost northern (southern) boundary amoebot that is not on $G$ or $G'$ (see \Cref{lem:tp:unmark}).

If they do not exist, i.e., $\operatorname{portal}_q(G)$ and $\operatorname{portal}_q(G')$ are disjoint, we compute $\gate{G}{q}$ and $\gate{G'}{q}$ as follows.
Consider the $q$-portal graph $\mathcal P_q$ of the region.
Since the tunnel region is simple, $\mathcal P_q$ is a tree.
Let $\mathcal Q_q = \operatorname{portal}_q(G \cup G')$.
Let $R_q \in \mathcal Q_q$ be the northernmost portal in $\operatorname{portal}_q(G)$.
We apply the root and prune primitive on $\mathcal P_q$ with $\mathcal Q_q$ and $R_q$ (see \Cref{th:tp:rap}).
This allows $\gate{G}{q}$ ($\gate{G'}{q}$) to identify itself since it is the only gate in $\operatorname{portal}_q(G)$ ($\operatorname{portal}_q(G')$) with a neighbor not in $\operatorname{portal}_q(G)$ ($\operatorname{portal}_q(G')$).
$\gate{G}{q}$ ($\gate{G'}{q}$) identifies $\boundarynode{G}{q}$ ($\boundarynode{G'}{q}$) by determining its westernmost boundary amoebot that is not on $G$ ($G'$) (see \Cref{lem:tp:unmark}).

Next, the whole tunnel region checks whether $P_{\uparrow,x}$/$P_{\downarrow,x}$, or $P_{\uparrow,z}$/$P_{\downarrow,z}$ exist (see \Cref{lem:tp:property:region}).
If at least one pair of portals exists, Case 1 holds for $q = x$, $q = z$, or both so that we terminate (see \Cref{fig:convex_decomposition}).
Otherwise, Case 2 holds for both, $q = x$ and $q = z$, and we proceed as follows.
First, $M$ identifies itself by checking whether it is intersected by both, $\gate{G}{x} \cup \gate{G}{z}$ and $\gate{G'}{x} \cup \gate{G'}{z}$ (see \Cref{lem:tp:property:region}).
Then, $g$ ($g'$) identifies itself.
It is either the intersection of $\gate{G}{x}$ and $\gate{G}{z}$ ($\gate{G'}{x}$ and $\gate{G'}{z}$), $\boundarynode{G}{x}$ ($\boundarynode{G'}{x}$), or $\boundarynode{G}{z}$ ($\boundarynode{G'}{z}$).
Note that only one of these cases can hold at the same time.

For $q \in \{ x, y, z \}$, we apply the root and prune algorithm on $\mathcal P_q$ with $\mathcal Q = \{ \operatorname{portal}_q(g), \operatorname{portal}_q(g') \}$ and $R = \operatorname{portal}_q(g)$ (see \Cref{th:tp:rap}).
Note that since $|\mathcal Q| = 2$, we obtain a path of portals.
We get $d_q = d_q(\operatorname{portal}_q(g),d_q(\operatorname{portal}_q(g'))) = \max_{P \in \mathcal P_q} d(\operatorname{portal}_q(g),P) = \max_{P \in \mathcal P_q} d(\operatorname{portal}_q(g'),P)$.
Hence, we use the PASC algorithm with $\operatorname{portal}_q(g)$ ($\operatorname{portal}_q(g')$) as the root to compute the values $d_q$, $d_q(\operatorname{portal}_q(g),P)$ and $d_q(\operatorname{portal}_q(g'),P)$ for each portal $P$ (see \Cref{th:tp:pasc}).
Each portal $P$ compares their distances to $d_q/2$.
This allows portal $P_q$ to identify itself.

Finally, let $\mathcal P'_q = (V_{\mathcal P'_q},E_{\mathcal P'_q})$ denote the pruned portal tree.
Then, $S_M = \bigcap_{q \in \{ x, y, z \}}\bigcup_{P \in V_{\mathcal P'_q}}P$ denotes the set of all amoebots that are on a shortest path between $g$ and $g'$ \cite{DBLP:conf/podc/PadalkinS24}.
Note that each amoebot can determine whether it is in $S_M$ since it knows whether it is on a portal of $\mathcal P'_q$.
It can also determine whether its removal would disconnect $S_M$, which implies that all shortest paths between $g$ and $g'$ within $M$ go through it.
Let $B_q$ denote the set of all these amoebots.
For $q \in \{ x, y, z \}$, we check whether $R_q(g) = R_q(g')$ (see \Cref{lem:tp:property:region}).
If this is the case, portal $P_q$ identifies amoebot $b_q$ by computing the westernmost amoebot of $B_q \cap P_q$ (see \Cref{lem:tp:unmark}).

\begin{lemma}
\label{lem:amoebot:convex}
    Given a decomposition consisting of $\Theta(|\mathcal{H}|)$ simple tunnel regions,
    an amoebot structure computes a decomposition consisting of $\Theta(|\mathcal{H}|)$ simple convex regions within $O(\log n)$ rounds w.h.p.
\end{lemma}

\begin{proof}
    The correctness follows from \Cref{lem:path_convex_regions_correctness}, and the runtime follows from \Cref{th:leader_election,th:tp:pasc,th:tp:rap,lem:tp:property:region,lem:tp:unmark}.
\hfill $\square$ 
\end{proof}

Combining \Cref{lem:amoebot:simple,lem:amoebot:tunnel,lem:amoebot:convex} proves \Cref{th:main_theorem}.



\section{Conclusion and Future Work}
\label{sec:conclusion}

We have shown how to decompose general regular triangular grid graphs into simple, geodesically convex regions and how to compute such a decomposition in the amoebot model with reconfigurable circuits in logarithmic time.
Notably, the decomposition is model-independent and might be of broader interest beyond the amoebot model.
While our focus has been on the decomposition problem itself, prior work suggests that a wide range of problems can benefit from such decompositions (e.g. motion planning).
Exploring these applications remains an interesting direction for future work.

\subsubsection*{Acknowledgements.} 
We thank Dona Davis, Anns Mary Francis, and Alex Poovathummoottil Sibichen for their contributions to an early draft of the paper.



\bibliography{references}

\end{document}